\documentclass[a4paper,UKenglish,cleveref, autoref, thm-restate]{lipics-v2021}

\usepackage{scalefnt,cite}
\usepackage{xspace}
\usepackage{amsmath}
\usepackage{graphicx,url}
\usepackage{booktabs}
\usepackage{float}
\usepackage{algpseudocode}
\usepackage{soul}

\usepackage[english,algoruled,longend,ruled,vlined]{algorithm2e}

\floatname{algorithm}{Listing}

\usepackage{amssymb}
\usepackage{bm}
\usepackage{dsfont}
\usepackage{hf-tikz}
\usepackage{marvosym}

\usepackage{todonotes,thm-restate}

\newenvironment{cproof}{\proof[Proof of claim]}{\endproof}

\definecolor{RED}{RGB}{255,0,0}

\newcommand{\NP}{\mathcal{NP}}

\newcommand{\tmp}{\mathcal{\sf TMP}}

\newcommand{\parent}{{\sf parent}\xspace}
\newcommand{\child}{{\sf children}\xspace}
\renewcommand{\root}{{\sf root}\xspace}
\newcommand{\anc}{{\sf ancestors}\xspace}
\newcommand{\descendants}{{\sf desc}\xspace}
\newcommand{\rot}{{\sf rot}\xspace}
\newcommand{\dist}{{\sf dist}\xspace}
\newcommand{\diam}{{\sf diam}\xspace}
\newcommand{\Bcb}{B_{\sf cb}\xspace}
\newcommand{\trace}{{\sf trace}\xspace}
\newcommand{\typechild}{{\sf type\text{-}children}\xspace}

\newcommand{\wantparent}{{\sf {\sf want{\text -}parent}}\xspace}

\newcommand{\Ocal}{{\mathcal O}\xspace}

\usepackage{color}

\usepackage{tikz}
 \usetikzlibrary{calc}
 \usepackage{boxedminipage}
\usepackage{framed}
\usepackage{xspace}
\usepackage{xifthen}
 \usepackage{tabularx}

\newlength{\RoundedBoxWidth}
\newsavebox{\GrayRoundedBox}
\newenvironment{GrayBox}[1]%
   {\setlength{\RoundedBoxWidth}{.93\textwidth}
    \def\boxheading{#1}
    \begin{lrbox}{\GrayRoundedBox}
       \begin{minipage}{\RoundedBoxWidth}}%
   {   \end{minipage}
    \end{lrbox}
    \begin{center}
    \begin{tikzpicture}%
       \node(Text)[draw=black!20,fill=white,rounded corners,%
             inner sep=2ex,text width=\RoundedBoxWidth]%
             {\usebox{\GrayRoundedBox}};
        \coordinate(x) at (current bounding box.north west);
        \node [draw=white,rectangle,inner sep=3pt,anchor=north west,fill=white]
        at ($(x)+(6pt,.75em)$) {\boxheading};
    \end{tikzpicture}
    \end{center}}

\newenvironment{defproblemx}[2][]{\noindent\ignorespaces%
                                \FrameSep=6pt%
                                \parindent=0pt%
                \vspace*{-1.5em}
                \ifthenelse{\isempty{#1}}{%
                  \begin{GrayBox}{\textsc{#2}}%
                }{%
                  \begin{GrayBox}{\textsc{#2} parameterized by~{#1}}%
                }
                \begin{tabular*}{\textwidth}{@{\hspace{.1em}} >{\itshape} p{1.8cm} p{0.8\textwidth} @{}}%
            }{
                \end{tabular*}%
                \end{GrayBox}%
                \ignorespacesafterend
            }

\newcommand{\defproblema}[3]{%
  \begin{defproblemx}{#1}
    {\bf Instance:}  & #2 \\
    {\bf Question:} & #3
  \end{defproblemx}
}%

\newcommand{\defproblemaparam}[4]{%
  \begin{defproblemx}{#1}
    {\bf Instance:}  & #2 \\
    {\bf Parameter:}  & #3 \\
    {\bf Question:} & #4
  \end{defproblemx}
}%

\newtheorem*{customprop*}{Proposition}
\newtheorem*{customlmm*}{Lemma}
\newtheorem*{customdef*}{Definition}
\newtheorem*{customthm*}{Theorem}
\newtheorem*{customclaim*}{Claim}
\newtheorem*{customcor*}{Corollary}
\newtheorem*{customobs*}{Observation}

 \hypersetup{
     colorlinks, linkcolor={dark-blue},
    citecolor={mid-green}, urlcolor={dark-blue}}

  \definecolor{mid-green}{rgb}{0.15,0.65,0.15}
 \definecolor{dark-green}{rgb}{0.15,0.25,0.15}
 \definecolor{dark-red}{rgb}{0.7,0.15,0.15}
 \definecolor{dark-blue}{rgb}{0.15,0.15,0.9}
 \definecolor{medium-blue}{rgb}{0,0,0.5}
 \definecolor{gray}{rgb}{0.5,0.5,0.5}
 \definecolor{color-Ig}{rgb}{0.15,0.7,0.15}
 \definecolor{darkmagenta}{rgb}{0.30, 0.0, 0.30}
 \definecolor{blue}{rgb}{0.15,0.15,0.9}

\newcommand{\ig}[1]{\textcolor{red}{[Ig: #1]}}

\renewcommand{\NP}{{\sf NP}\xspace}
\renewcommand{\P}{{\sf P}\xspace}
\newcommand{\FPT}{{\sf FPT}\xspace}

\newcommand{\kelimination}{{\sc Rotation Distance}\xspace}

\title{Computing Distances on Graph Associahedra is Fixed-parameter Tractable}
\title{Computing distances is FPT on graph associahedra and W[2]-hard on hypergraphic polytopes}

\author{Lu\'is Felipe I. Cunha}{Instituto de Computação, Universidade Federal Fluminense, Brasil \and \url{http://www.ic.uff.br/~lfignacio} }{lfignacio@ic.uff.br}{https://orcid.org/0000-0002-3797-6053}{FAPERJ-JCNE (E-26/201.372/2022) and~CNPq-Universal~(406173/2021-4).}

\author{Ignasi Sau}{LIRMM, Université de Montpellier, CNRS, France \and \url{https://www.lirmm.fr/~sau/} }{ignasi.sau@lirmm.fr}{https://orcid.org/0000-0002-8981-9287}{French project \textsc{ELiT} (ANR-20-CE48-0008-01).}

\author{U\'everton S. Souza}{Instituto de Computação, Universidade Federal Fluminense, Brasil \and IMPA - Instituto de Matem\'atica Pura e Aplicada, Brasil \and \url{http://www.ic.uff.br/~ueverton} }{ueverton@ic.uff.br}{https://orcid.org/0000-0002-5320-9209}{CNPq (312344/2023-6), and FAPERJ (E-26/201.344/2021).}

\author{Mario Valencia-Pabon}{Université de Lorraine, CNRS, Inria, LORIA, F-54000 Nancy, France \and \url{https://lipn.univ-paris13.fr/~valenciapabon/}}
{mario.valencia@loria.fr}{https://orcid.org/0009-0006-0564-4341}{French project \textsc{Abysm} (ANR-23-CE48-0017).}

\authorrunning{L. Cunha, I. Sau, U. S. Souza, and M. Valencia-Pabon}

\Copyright{Cunha, Sau, Souza, and Valencia-Pabon}



\relatedversion{An extended abstract presenting the \FPT algorithm on graph associahedra appeared in the proceedings of \emph{ICALP 2025}.}

\hideLIPIcs
\nolinenumbers

\begin{document}

\maketitle

\begin{abstract}
An elimination tree of a connected graph $G$ is a rooted tree on the vertices of $G$ obtained by choosing a root~$v$ and recursing on the connected components of $G-v$ to obtain the subtrees of $v$. The graph associahedron of $G$ is a polytope whose vertices correspond to elimination trees of $G$ and whose edges correspond to tree rotations, a natural operation between elimination trees. These objects generalize associahedra, which correspond to the case where $G$ is a path. Ito et al.~[SIAM J. Discrete Math. 2026] recently proved that the problem of computing distances on graph associahedra is \NP-hard. In this paper we prove that the problem, for a general graph $G$,  is fixed-parameter tractable (\FPT) parameterized by the distance~$k$. Prior to our work, only the case where $G$ is a path was known to be \FPT. To prove our result, we use a novel approach based on a marking scheme that restricts the search to a set of vertices whose size is bounded by a (large) function of $k$.

On the negative side, we  show that it is unlikely that \FPT algorithms exist on a natural generalization of graph associahedra, namely hypergraphic polytopes, by proving that computing distances on them is ${\sf W}[2]$-hard parameterized by the distance. We also prove that, on hypergraphic polytopes, the distance cannot be approximated in polynomial time within a factor $c \cdot \log(|V|+|\mathcal{E}|)$ for some constant $c > 0$ unless $\P = \NP$, where $H=(V, \mathcal{E})$ is the input hypergraph. This result strengthens the hardness result of Cardinal and Steiner [Combin. Theory 2025], who proved that the problem cannot be approximated within a factor $(1 + \varepsilon)$ for some absolute constant $\varepsilon > 0$ unless $\P = \NP$. Finally, we rule out the existence of polynomial kernels parameterized by the number of vertices of the input hypergraph, a parameter for which the problem is easily seen to be \FPT.
\end{abstract}

\newpage

\section{Introduction}
\label{sec:intro}


Given a connected and undirected graph $G$, an \emph{elimination tree} $T$ of $G$ is any rooted tree that can be defined recursively as follows. If $V(G)=\{v\}$, then $T$ consists of a single root vertex $v$. Otherwise, a vertex $v \in V(G)$ is chosen as the root of $T$, and an elimination tree is created for each connected component of $G - v$. Each root of these elimination trees of $G - v$ is a child of $v$ in $T$. For a disconnected graph $G$, an \textit{elimination forest} of $G$ is the disjoint union of elimination trees of the connected components of $G$. Equivalently, an elimination forest of a graph $G$ is a rooted
forest $F$ (that is, a forest with a root in every connected component) on vertex set $V(G)$ such that for each edge $u v \in E(G)$, vertex $u$ is an ancestor
of vertex $v$ in $F$, or vice versa.

\autoref{fig:eliminationtree} illustrates an example of two elimination trees $T$ and $T'$ of a graph $G$. With slight (and standard) abuse of notation, we use the same labels for the vertices of a graph $G$ and any of its elimination trees. Note that an elimination tree is unordered, i.e., there is no ordering associated with the children of a vertex in the tree. Similarly, there is no ordering among the elimination trees in an elimination forest.


\begin{figure}[htb]
    \centering
    \vspace{-.15cm}
    \includegraphics[width=.75\textwidth]{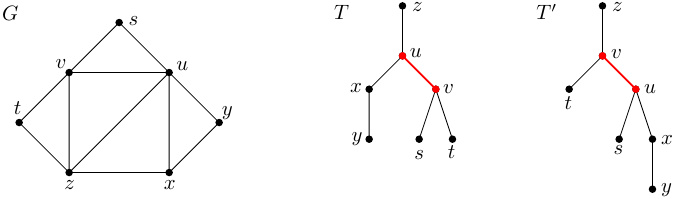}
    \caption{A graph $G$ and two of its elimination trees $T$ and $T'$, where the second one is obtained from the first one by the rotation of edge $uv$ (in red).
    \label{fig:eliminationtree}}
\end{figure}



Elimination trees have been studied extensively in various contexts, including graph theory, combinatorial optimization, polyhedral combinatorics, data structures, or VLSI design; see the recent paper by Cardinal, Merino, and M{\"u}tze~\cite{cardinal2022efficient} and the references therein. In particular, elimination trees play a prominent role in structural and algorithmic graph theory, as they appear naturally in several contexts. As a relevant example, the \textit{treedepth} of a graph $G$ is defined as the minimum height of an elimination forest of $G$~\cite{books/daglib/0030491}. 



Given a class of combinatorial objects and a ``local change'' operation between them, the corresponding \emph{flip graph} has as vertices the combinatorial objects, and its edges connect pairs of objects that differ by the prescribed change operation. In this article, we focus on the case where this class of combinatorial objects is the set of elimination forests of a graph~$G$. For these objects,  the  commonly considered  ``local change'' operation is that of \emph{edge rotation} defined as follows, where we suppose for simplicity that $G$ is connected.  Given an elimination tree $T$ of a graph $G$, the \emph{rotation} of an edge $uv \in E(T)$, with $u$ being the parent of $v$, creates another elimination of $G$, denoted by $\rot(T,uv)$, obtained, informally, by just swapping the choice of $u$ and $v$ in the recursive definition of $T$ (that is, in the so-called \emph{elimination ordering}), and updating the parent of the subtrees rooted at $v$ accordingly;  see \autoref{fig:rotation1} for an illustration. The formal definition can be found in \autoref{sec:prelim} (cf. \autoref{def:rotation}).





\begin{figure}[h!tb]
    \centering
    \vspace{-.15cm}
    \includegraphics[width=.77\textwidth]{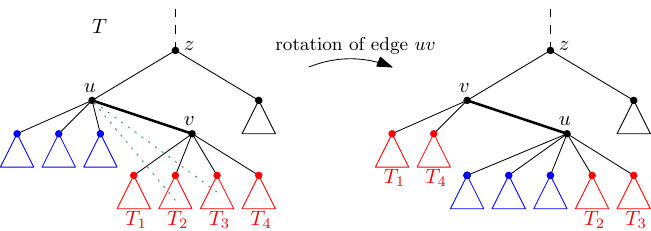}
    \caption{
    On the left:
    An elimination tree $T$ of a graph $G$ with adjacent vertices $u$ and $v$. Vertex~$v$ has four subtrees, and two of them, namely $T_2$ and $T_3$, contain vertices adjacent to vertex $u$ in~$G$.
    On the right:
    Elimination tree resulting from $T$ by applying the rotation of $uv$.
    Since both $G[V(T_2) \cup \{u\}]$ and $G[V(T_3) \cup \{u\}]$ are connected, $T_2$ and $T_3$ become subtrees of $u$ in $\rot(T,uv)$.
    \label{fig:rotation1}}
\end{figure}


For example, in \autoref{fig:eliminationtree}, $T'=\rot(T,uv)$. The definition of the rotation operation clearly implies that it is self-inverse with respect to any edge, that is, for any elimination tree $T$ of a graph $G$ and any edge $uv \in E(T)$, it holds that $T= \rot(\rot(T,uv),vu)$. The \emph{rotation distance} between two elimination trees $T,T'$ of a graph $G$, denoted by $\dist(T,T')$, is the minimum number of rotations it takes to transform $T$ into $T'$. The self-invertibility property of rotations discussed above implies that  $\dist(T,T') = \dist(T',T)$.




It is well known that for any graph $G$, the flip graph of elimination forests of $G$ under tree rotations is the skeleton of a polytope, referred to as the \emph{graph associahedron} ${\mathcal A}(G)$ and that was introduced by Carr, Devadoss, and Postnikov~\cite{Devadoss09,CarrD06,Postnikov09}. For the particular cases of $G$ being a complete graph, a cycle, a path,   a star, or a disjoint union of edges, ${\mathcal A}(G)$ is the permutahedron, the cyclohedron, the (standard) associahedron,  the stellohedron, or the hypercube, respectively; see the introduction of~\cite{cardinal2022efficient} for nice figures to illustrate these objects.

Graph associahedra naturally generalize
associahedra, which correspond to the particular case where $G$ is a path. As mentioned in~\cite{cardinal2022efficient},  the associahedron has a rich history and literature, connecting computer science, combinatorics, algebra, and topology~\cite{SleatorTT88,Loday04,HohlwegL07,Pournin14}. See the introduction of the paper by Ceballos, Santos,
and Ziegler~\cite{CeballosSZ15} for
a historical account. In an associahedron, each vertex corresponds to a binary tree over a set of $n$ elements, and each edge corresponds to a rotation operation between two binary trees, an operation used in standard online binary search tree algorithms~\cite{AVL-trees,GuibasS78,SleatorT85}. Binary trees are in bijection with many other Catalan objects such as triangulations of a convex polygon, well-formed parenthesis expressions, Dyck paths, etc.~\cite{Stanley15}. For instance, in triangulations of a convex polygon, the rotation operation maps to another simple operation, known as a flip, which removes the diagonal of a convex quadrilateral formed by two triangles and replaces it by the other diagonal.

\subparagraph*{Related work on graph associahedra.}  Distances on graph associahedra have been object of intensive study. Probably, the most studied parameter is the diameter, that is, the maximum distance between two vertices of ${\mathcal A}(G)$. A number of influential articles either determine the diameter exactly, or provide lower and upper bounds, or asymptotic estimates, for the cases where the underlying graph $G$ is a path~\cite{SleatorTT88,Pournin14}, a
star~\cite{MannevilleP10}, a cycle~\cite{Pournin17}, a tree~\cite{MannevilleP10,CardinalLP18}, a complete bipartite graph~\cite{CardinalPVP22},
a caterpillar~\cite{Berendsohn22}, a
trivially perfect graph~\cite{CardinalPVP22},
a graph in which some width parameter (such as treedepth or treewidth) is bounded~\cite{CardinalPVP22}, or a general graph~\cite{MannevilleP10}.

Our focus is on the algorithmic problem of determining the distance between two vertices of ${\mathcal A}(G)$, or equivalently, determining the rotation distance between two given elimination trees of a graph $G$. There are very few cases where this problem is known to be solvable in polynomial time, namely when $G$ is a complete graph (folklore), a star~\cite{CardinalPV24}, or a complete split graph~\cite{CardinalPV24}. The complexity of the case where $G$ is a path has been a notorious long-standing open problem for several decades, and it has been very recently settled by Dorfer~\cite{Dorfer26}, who has proved it to be \NP-complete. On the positive side, for $G$ being a path,  there exist a polynomial-time $2$-approximation algorithm~\cite{ClearyJ10} and several fixed-parameter tractable (\FPT)
algorithms when the distance is the parameter~\cite{cleary2009rotation,kanj2017computing,li20233,lubiw2015flip,Lucas10}.
It is worth mentioning that there are some hardness results on generalized settings~\cite{LubiwP15,Pilz14,AichholzerMP15} and polynomial-time algorithms for some type of restricted rotations~\cite{Cleary23}.

Cardinal et al.~\cite{cardinal4} asked whether computing distances on general graph associahedra is \NP-hard (before the hardness result of Dorfer~\cite{Dorfer26} was known). This question was answered positively by Ito et al.~\cite{ito_et_al:LIPIcs.ICALP.2023.82}. 

%
%

\subparagraph*{Our results.} The \NP-hardness result of Ito et al.~\cite{ito_et_al:LIPIcs.ICALP.2023.82} (see also~\cite{CardinalS25}) paves the way for studying the parameterized complexity of the problem of computing distances on graph associahedra. Thus, in this article we are interested in the following parameterized problem, where we consider the natural parameter, that is, the desired distance. 


\medskip
\defproblemaparam{\kelimination}
{A graph $G$, two elimination trees $T$ and $T'$ of $G$, and a positive integer $k$.}
{$k$.}
{Is the rotation distance between $T$ and $T'$ at most $k$?}


As mentioned above, \kelimination was known to be polynomial-time solvable on complete graphs, stars, and complete split graphs~\cite{CardinalPV24}, and \FPT algorithms were only known on paths~\cite{cleary2009rotation,kanj2017computing,li20233,lubiw2015flip,Lucas10}.
In this article we vastly generalize the known results by providing an \FPT algorithm to solve the \kelimination problem for a general input graph $G$. More precisely, we prove the following theorem.

\begin{restatable}{theorem}{maintheorem}\label{thm:main}
The  \kelimination problem can be solved in time $f(k) \cdot |V(G)|$, with
\vspace{-.25cm}
\begin{equation*}\label{eq:main-running-time}
    f(k) = k^{k \cdot 2^{2^{.^{.^{.{2^{\Ocal(k^2)}}}}}}}, \text{ where the tower of exponentials has height at most  $(3k+1)4k = \Ocal(k^2)$.}
\end{equation*}
\end{restatable}

In particular, \autoref{thm:main} yields a linear-time algorithm to solve \kelimination for every fixed value of the distance $k$. To the best of our knowledge, this is the first positive algorithmic result for the general \kelimination problem (i.e., with no restriction on the input graph $G$), and we hope that it will find algorithmic applications in the many contexts where graph associahedra arise naturally~\cite{Devadoss09,CarrD06,Postnikov09,cardinal2022efficient,MannevilleP10,ito_et_al:LIPIcs.ICALP.2023.82}.
Our result can also been seen through the lens of the very active area of the parameterized complexity of graph reconfiguration problems; see~\cite{BousquetMNS24} for a recent survey.


\medskip

It is natural to ask whether the \FPT algorithm of \autoref{thm:main} can be generalized beyond graph associahedra. In the second part of the article we prove that, for a natural generalization of graph associahedra called \emph{hypergraphic polytopes} (defined in \autoref{sec:hypergraphic}, see~\cite{CardinalS25,BBM19,Hel74,VW93,AA23,PRW08,Post09,Reh22} for more background), this is probably not the case, since we prove (cf.~\autoref{thm:w2hard}) that computing distances on hypergraphic polytopes is ${\sf W}[2]$-hard parameterized by the distance. We prove this by a parameterized reduction from \textsc{Dominating Set} parameterized by the size of the solution.

By appropriately modifying the ${\sf W}[2]$-hardness reduction, we also transfer the hardness of approximation from  \textsc{Dominating Set} to computing distances on hypergraphic polytopes. Namely, as it is the case of \textsc{Dominating Set}~\cite{RazS97}, we prove (cf.~\autoref{thm:inapproximability}) that
there exists some constant $c > 0$ such that computing distances on hypergraphic polytopes does not admit a polynomial-time $c \cdot \log (|V|+|{\mathcal E}|)$-approximation algorithm unless $\P = \NP$, where $H=(V,{\mathcal E})$ is the input hypergraph. This result strengthens the recent hardness result of Cardinal and Steiner~\cite{CardinalS25}, who proved that the problem cannot be approximated within a factor $(1 + \varepsilon)$ for some absolute constant $\varepsilon > 0$ unless $\P = \NP$.

Finally, another appropriate modification of the ${\sf W}[2]$-hardness reduction, namely by reducing from the \textsc{Red-Blue Dominating Set} problem instead, allows us (cf.~\autoref{thm:no-kernels}) to rule out the existence of polynomial kernels for computing distances on hypergraphic polytopes parameterized by the number of vertices of the input hypergraph, unless ${\sf NP} \subseteq {\sf coNP}/{\sf poly}$. As we discuss in \autoref{sec:no-kernels}, the problem is easily seen to be \FPT with this parameterization.

\subparagraph*{Organization.} In \autoref{sec:prelim} we provide standard preliminaries about graphs and parameterized complexity and fix our notation. In \autoref{sec:sketch} we present an overview of the main ideas of the algorithm of \autoref{thm:main}, which may serve as a road map to read the formal description of the \FPT algorithm presented in \autoref{sec:description-algorithm}, split into several subsections. We prove the hardness results in \autoref{sec:hardness}, and in \autoref{sec:discussion} we discuss several directions for further research, both on graph associahedra and hypergraphic polytopes.


\section{Preliminaries}\label{sec:prelim}

\subparagraph*{Graphs.} We use standard graph-theoretic notation, and we refer the reader to~\cite{Diestel16} for any undefined terms. An edge between two vertices $u,v$ of a graph $G$ is denoted by $uv$. For a graph $G$ and a vertex set $S \subseteq V(G)$, the graph $G[S]$ has vertex set $S$ and edge set $\{uv \mid u,v \in S \text{ and }uv \in E(G)\}$. A \emph{connected component} $Z$ of a graph $G$ is a connected subgraph that is maximal (with respect to the addition of vertices or edges) with this property. We let ${\sf cc}(G)$ denote the set of connected components of a graph $G$. The \emph{distance} between two vertices $x,y$ in $G$, denoted by $\dist_G(x,y)$, is the length of a shortest path between $x$ and $y$ in $G$. The \emph{diameter} of $G$, denoted by $\diam(G)$, is the maximum length of a shortest path between any two vertices of $G$. We will often consider distances and the diameter of some rooted tree $T$ that is (a subtree of) an elimination tree of a graph $G$. We stress that $\dist_T(x,y)$ refers to the distance between $x$ and $y$ in $T$, not in $G$, and the same applies to $\diam(T)$.

For a graph $G$, a vertex $v \in V(G)$, and an integer $r \geq 1$, we denote by $N_G^{r}[v]$ the set of vertices within distance at most $r$ from $v$ in $G$, including $v$ itself. For a set $S \subseteq V(G)$, we let $N_G^{r}[S] = \bigcup_{v \in S}N_G^{r}[v]$. For a subgraph $H$ of $G$, we use $N_G^{r}(H)$ as a shortcut for $N_G^{r}(V(H))$. In all these notations, we omit the superscript $r$ in the case where $r=1$, that is, to refer to the usual neighborhood.

For a positive integer $p$, we let $[p]$ denote the set $\{1, 2, \ldots, p\}$. If $f:A \to B$ is a function between two sets $A$ and $B$ and $A' \subseteq A$, we denote by $f|_{A'}$ the restriction of $f$ to $A'$.

\subparagraph*{Rooted trees.} For a rooted tree $T$, we use $\root(T)$ to denote its root. For a vertex $v \in V(T)$, we denote by $\parent(T,v)$ the unique parent of $v$ in $T$ (or the empty set if $v$ is the root),  by $\child(T,v)$ the set of children of $v$ in $T$, by $\anc(T,v)$ the set of ancestors of $v$ in $T$ (including $v$ itself), and by $\descendants(T,v)$ the set of descendants of $v$ in $T$ (including $v$ itself). The \emph{strict} ancestors (resp. descendants) of $v$ are the vertices in the set $\anc(T,v) \setminus \{v\}$ (resp. $\descendants(T,v) \setminus \{v\}$). We denote by $T(v)$ the subtree of $T$ rooted at $v$. Two vertices $v,v' \in V(T)$ are $T$-\emph{siblings} if $\parent(T,v) = \parent(T,v')$.

\subparagraph*{Rotation of an edge in an elimination tree.} We provide the formal definition of the rotation operation, which has been already informally defined in the introduction (cf. \autoref{fig:rotation1}).

\begin{definition}[rotation operation]\label{def:rotation}
Let $T$ be an elimination tree of a graph $G$ and let $uv \in E(T)$ with $\parent(T,v) = u$. The {\em rotation of $uv$} in $T$ creates another elimination tree $\rot(T,uv)$ of $G$ defined as follows, where for better readability we use $T' = \rot(T,uv)$:
\begin{enumerate}
    \item $\parent(T',u)=v$.
    \item $u \in \child(T',v)$.
    \item If $u \neq \root(T)$, let $z = \parent(T,u)$. Then $\child(T',z) =(\child(T,z) \setminus\{u\}) \cup \{v\}$.
    \item $\child(T,u) \subseteq \child(T',u)$.
    \item Let $w \in \child(T,v)$. If $u$ is adjacent in $G$ to some vertex in $T(w)$, then $w \in \child(T',u)$; otherwise $w \in \child(T',v)$.
    \item\label{item6} For every vertex $s \in V(G) \setminus \{u,v,z\}$, $\child(T',s)=\child(T,s)$.
\end{enumerate}
\end{definition}

A \textit{$k$-rotation sequence} from an elimination tree $T$ to another elimination tree $T'$ (of the same graph $G$) is an ordered set $(e_1,\ldots,e_k)$ of edges such that, letting inductively $T_0 := T$ and, for $i \in [k]$, $T_i:= \rot(T_{i-1},e_i)$ with $e_i \in E(T_{i-1})$, we have that $T_k=T'$. In other words, a $k$-rotation sequence consists of the ordered list of the $k$ edges to be rotated in order to obtain $T'$ from $T$, going through the intermediate elimination trees $T_1, \ldots, T_{k-1}$ (of the same graph $G$). Clearly, $\dist(T,T') \leq k$ if and only if there exists an $\ell$-rotation sequence from $T$ to $T'$ for some $\ell \leq k$. We say that a vertex $v \in V(T)$ is \emph{used} by a rotation sequence $\sigma$ if it is an endpoint of some of the edges that are rotated by $\sigma$.

\subparagraph*{Parameterized complexity.}  A \emph{parameterized problem} is a language $L \subseteq \Sigma^* \times \mathbb{N}$, for some finite alphabet $\Sigma$.  For an instance $(x,k) \in \Sigma^* \times \mathbb{N}$, the value~$k$ is called the \emph{parameter}. Such a problem is \emph{fixed-parameter tractable} (\FPT for short) if there is an algorithm that decides membership of an instance $(x, k)$ in time $f(k)\cdot {|x|}^{\Ocal(1)}$ for some computable function~$f$.

For a computable function~$g \colon \mathbb{N} \to \mathbb{N}$, a \emph{kernelization algorithm} (or simply a \emph{kernel}) for a parameterized problem $L$ of \emph{size} $g$ is an algorithm $A$ that given any instance $(x,k)$ of $L$, runs in polynomial time and returns an instance $(x',k')$ such that $(x,k) \in L \Leftrightarrow (x', k') \in L$ with $|x'|$, $k' \le g(k)$. The function $g(k)$ is called the \textit{size} of the kernel, and a kernel is \textit{polynomial} (resp. \textit{linear, quadratic}) if $g(k)$ is a polynomial (resp. linear, quadratic) function.

To transfer the non-existence of polynomial kernels (under reasonable complexity assumptions), we use the notion of \emph{polynomial parameter transformation} (PPT for short), introduced by Bodlaender, Thomass{\'{e}}, and Yeo~\cite{BodlaenderTY11}.  A polynomial parameter transformation from a parameterized problem $P$ to a parameterized problem $Q$ is an algorithm that, given an instance $(x,k)$ of $P$, computes in polynomial time an equivalent instance $(x',k')$ of $Q$ such that $k'$ is bounded by a polynomial depending only on $k$. It follows easily from the definition that if $P$ does not admit a polynomial kernel under some complexity assumption, then the same holds for $Q$.

Within parameterized problems, the classes ${\sf W}[i]$ with $i\geq 1$ may be seen as the parameterized equivalent to the class {\sf NP} of classical decision problems. Without entering into details (see~\cite{DoFe13,CyganFKLMPPS15} for the formal definitions), a parameterized problem being {\sf W}[i]-\emph{hard} for some $i\geq 1$ can be seen as a strong evidence that this problem is {\sl not} {\sf FPT}.
The canonical example of ${\sf W}[1]$-hard (resp. ${\sf W}[2]$-hard) problem is \textsc{Clique} (resp. \textsc{Dominating Set}), both  parameterized by the size of the solution

Consult~\cite{CyganFKLMPPS15,Niedermeier06,FlumG06,DoFe13,FominLSZ19} for background on parameterized complexity.

\section{Overview of the main ideas of the algorithm}
\label{sec:sketch}

Our approach to obtain an \FPT algorithm to solve \kelimination is novel, and does not build on previous work. Given two elimination trees $T$ and $T'$ of a connected graph $G$ and a positive integer $k$, our goal is to decide whether there exists what we call an   \textit{$\ell$-rotation sequence} $\sigma$ from $T$ to $T'$, for some $\ell \leq k$, that is, an ordered list of $\ell$ edges to be rotated in order to obtain $T'$ from $T$, going through the intermediate elimination trees $T_1, \ldots, T_{\ell-1}$ (all of the same graph $G$); see \autoref{sec:prelim} for the formal definition. At a high level, our approach is based on identifying a subset of \emph{marked vertices} $M \subseteq V(T)$, of size bounded by a function of $k$, so that we can assume that the desired rotation sequence $\sigma$ uses only vertices in $M$. Once this is proved, an \FPT algorithm follows directly by applying brute force and guessing all possible rotations using vertices in $M$.

A crucial observation (cf.~\autoref{obs:ChangingChildren}) is that a rotation may change the set of children of at most three vertices (but the parent of arbitrarily many vertices, such as the roots of $T_2$ and $T_3$ in \autoref{fig:rotation1}). Motivated by this, we say that a vertex $v \in V(T)$ is \emph{$(T,T')$-children-bad} if its set of children in $T$ is different from its set of children in $T'$. By \autoref{obs:ChangingChildren}, we may assume (cf.~\autoref{prop:+3kno}) that we are dealing with an instance in which the number of $(T,T')$-children-bad vertices is at most $3k$.

In a first step, we prove (cf.~\autoref{lemma:restriction-to-balls}) that we can assume that the desired sequence $\sigma$ of at most $k$ rotations to transform $T$ into $T'$ uses only vertices lying in the union of the balls of radius $2k$ around $(T,T')$-children-bad vertices of $T$, which we denote by $\Bcb$. The proof of \autoref{lemma:restriction-to-balls} exploits, in particular, the fact that a rotation may increase or decrease vertex distances (in the corresponding trees) by at most one (cf.~\autoref{eq:distance-changes-by-one}). This is then used to show that if a rotation uses some vertex outside of  $\Bcb$, then it can be ``simplified'' into another one that does not (cf. \autoref{fig:bounded-diam}).

By \autoref{lemma:restriction-to-balls}, we restrict henceforth to rotations using only vertices in $\Bcb$. We can consider each connected component $Z$ of $T[\Bcb]$, since it can be easily seen that we can assume that there are at most $k$ of them. By definition of $\Bcb$, the diameter of such a component $Z$ is $\Ocal(k^2)$ (cf. \autoref{eq:bounded-diameter}). Thus, the ``only'' obstacle to obtain the desired \FPT algorithm is that the vertices in $\Bcb$ can have an arbitrarily large degree. Note that in the particular case where the underlying graph $G$ has bounded degree, the maximum degree of any elimination tree of $G$ is bounded, and therefore in that case $|\Bcb|$ is bounded by a function of $k$, and  an \FPT algorithm follows immediately. To the best of our knowledge, this result was not known for graphs other than paths (albeit, with a better running time than the one that results from just brute-forcing on the set $\Bcb$, which is of the form $2^{2^{\Ocal(k)}} \cdot |V(G)|$).

Our strategy to deal with high-degree vertices in $\Bcb$ is as follows. Fix one connected component $Z$ of $T[\Bcb]$. Our goal is to identify a subset $M_Z \subseteq V(Z)$ of size bounded by a function of $k$, such that we can restrict our search to rotations using only vertices in $M_Z$. To find such a ``small'' set $M_Z \subseteq V(Z)$, we define the notion of \textit{type} of a vertex $v \in V(Z)$, in such a way that the number of different types is bounded by a function of $k$.  Then, we will prove via our marking algorithm that it is enough to keep in $M_Z$, for each type, a number of vertices bounded again by a function of $k$. 

Before defining the types, we need to define the \textit{trace} of a vertex $v$ in $Z$. To get some intuition, look at the rotation depicted in \autoref{fig:rotation1}. Note that, for each of the subtrees $T_1,\ldots,T_4$ that are children of $v$ in $T$, what determines whether they are children of $u$ or $v$ in the resulting subtree is whether some vertex in $T_i$ is adjacent to $u$ or not. Iterating this idea, if we are about to perform at most $k$ rotations starting from $T$, then the behavior of such a subtree $T_i$, assuming that no vertex of it is used by a rotation, is determined by its neighborhood in a set of ancestors of size at most the diameter of $Z$, and this is what the trace is intended to represent. That is, the trace of a vertex $v$ in $Z$, denoted by $\trace(T,Z,v)$, captures ``abstractly'' the neighborhood of the whole subtree rooted at $v$ among (the ordered set of) its ancestors within the designated vertex set $Z \subseteq V(T)$; see \autoref{def:trace} for the formal definition of trace and \autoref{fig:example-trace} for an example. We stress that, when considering the neighborhood in the set of ancestors,  we look at the whole subtree $T(v)$ rooted at $v$, and not only at its restriction to the set $Z$.

Equipped with the definition of trace, we can define the notion of type, which is somehow involved (cf. \autoref{def:type}) and whose intuition behind is the following. For our marking algorithm to make sense, we want that if two vertices $v,v'$ with the same parent (called $T$-siblings) have the same type (within $Z$), denoted by $\tau(T,Z,v)=\tau(T,Z,v')$, and an $\ell$-rotation sequence $\sigma$ from $T$ to $T'$ uses some vertex from $T(v)$  but uses no vertex in $T(v')$, then there exists another $\ell$-rotation sequence $\sigma'$ from $T$ to $T'$ that uses vertices in $T(v')$ instead of those in $T(v)$. To guarantee this replacement property, we need a stronger condition than just $v$ and $v'$ having the same trace. Informally, we need them to have the same ``variety of traces among their children within $Z$''. More formally, this leads to a recursive definition where, in the leaves of $Z$ (that are not necessarily leaves of $T$), the type corresponds to the trace, and for non-leaves, the type is defined by the trace and by the number of children of each possible lower type. Note that, a priori, the number of children of a given type may be unbounded, which would rule out the objective of bounding the number of types as a function of $k$. To overcome this obstacle, the crucial and simple observation is that at most $k$ subtrees rooted at a vertex of $T$ contain vertices used by the desired rotation sequence $\sigma$ (cf. \autoref{lem:few-siblings-involved}). This implies that if there are at least $k+1$ $T$-siblings of the same type, necessarily the whole subtree of at least one of them, say $u$, will not be used by $\sigma$, implying that $u$ (and its whole subtree) achieves the desired parent in $T'$ without being used by $\sigma$, and the same occurs to any other $T$-sibling of the same type. Thus, keeping track of the existence of at least $k+1$ such children (regardless of their actual number) is enough to capture this
``static'' behavior, and allows us to shrink the possible distinct numbers to keep track of to a function of $k$ (cf.~\autoref{eq:definition-type}, where the ``$\min$'' is justified by the previous discussion). Finally, for technical reasons we also incorporate into the type of a vertex its desired parent in $T'$, in case it defers from its parent in $T'$ (cf. function $\wantparent(T,T',\cdot)$). See \autoref{def:type} for the formal definition of type and \autoref{fig:example-types} for an example with $k=2$.

We prove (cf.~\autoref{lem:number-types}) that the number of types is indeed bounded by a (large) function $g(k)$ depending only on $k$, and this function is what yields the upper bound on the asymptotic running time of the \FPT algorithm of \autoref{thm:main}. Moreover, we show (cf. \autoref{obs:computation-types}) that the type of a vertex can be computed in time $g(k) \cdot |V(G)|$. We then use the notion of type and the bound given by \autoref{lem:number-types} to define the desired set
$M_Z \subseteq Z$ of size bounded by a function of $k$.
In order to find $M_Z$, we apply a marking algorithm on $Z$, that first identifies a set $M_Z^{\sf pre} \subseteq V(Z)$ of \textit{pre-marked} vertices, whose size is not necessarily bounded by a function of $k$, and then ``prunes'' this set $M_Z^{\sf pre}$ in a root-to-leaf fashion to find the desired  set of \textit{marked} vertices $M_Z \subseteq M_Z^{\sf pre}$ of appropriate size. See \autoref{fig:example-marking} for an example of the marking algorithm for $k=1$. We define $M = \cup_{Z \in {\sf cc}(T[\Bcb])}M_Z$ (where ${\sf cc}$ denotes the set of connected components), and we call it the set of \emph{marked vertices} of $T$. We prove (cf.~\autoref{lem:M-bounded}) that the size of $M$ is roughly equal to the number of types, and that the set $M$ can be computed in time \FPT.

Once we have our set of marked vertices $M$ at hand, it remains to prove that we can restrict the rotations to use only vertices in $M$. This is proved in our main technical result (cf.~\autoref{lem:main}), whose proof critically exploits the recursive definition of the types. In a nutshell, we consider an $\ell$-rotation sequence $\sigma$ from $T$ to $T'$, for some $\ell \leq k$, minimizing, among all $\ell$-rotation sequences from $T$ to $T'$, the number of used vertices in $V(T) \setminus M$. Our goal is to define another $\ell$-rotation sequence $\sigma'$ from $T$ to $T'$ using strictly less vertices in $V(T) \setminus M$ than $\sigma$, contradicting the choice of $\sigma$.  To this end, let $v \in V(T) \setminus M$ be a furthest (with respect to the distance to $\root(T)$) non-marked vertex of $T$ that is used by $\sigma$. We distinguish two cases.

In Case~1, we assume that $v$ has a marked $T$-sibling $v'$ with $\tau(T,Z,v)=\tau(T,Z,v')$ (cf. \autoref{fig:mainlemma-case1}). It is not difficult to prove that we can define $\sigma'$ from $\sigma$ by just replacing $v$ with $v'$ in all the rotations of $\sigma$ involving $v$ (cf. \autoref{claim:case1-well-defined} and
\autoref{claim:case1-good-rotation-sequence}).

In Case~2, all $T$-siblings $v'$ of $v$ with $\tau(T,Z,v)=\tau(T,Z,v')$, if any, are non-marked.
In this case, in order to define another $\ell$-rotation sequence $\sigma'$ from $T$ to $T'$ that uses more marked vertices than $\sigma$, we need to modify $\sigma$ in a more {\sl global} way than in Case~1. Namely, in order to define $\sigma'$, we need a more global (and involved) replacement, which we achieve via what we call a \emph{representative function} $\rho$. To define $\rho$, we first guarantee the existence of a very helpful vertex $v^{\star}$ that is a non-marked ancestor of $v$ having a marked $T$-sibling $v'$ of the same type such that no vertex in $T(v')$ is used by $\sigma$; see \autoref{claim:vertex-vstar-exists} and \autoref{fig:mainlemma-case2}. Exploiting the recursive definition of type, we then define our representative function $\rho$, mapping vertices used by $\sigma$ in $T(v^{\star})$ to vertices in $T(v')$ of the same type (cf.~\autoref{claim:representative-function}), and prove that we can define $\sigma'$ from $\sigma$ by replacing each vertex $v$ used by $\sigma$ in $T(v^{\star})$ by its image via $\rho$ in $T(v')$ in all the rotations of $\sigma$ involving $v$ (cf. \autoref{claim:case2-well-defined} and
\autoref{claim:case2-good-rotation-sequence}).

\section{Formal description of the \FPT algorithm}
\label{sec:description-algorithm}

In this section we present our \FPT algorithm to solve the \kelimination problem. We start in \autoref{sec:basic-observations} by providing some definitions and  useful  observations about the so-called \emph{good} and \emph{bad} vertices. In \autoref{sec:bounded-diameter} we show that we can assume that all the rotations involve vertices within balls of small radius around bad vertices. In \autoref{sec:description-marking} we describe our marking algorithm, using the definition of type, and show that the set of marked vertices can be computed in \FPT time. In \autoref{sec:restriction-to-marked} we prove our main technical result (\autoref{lem:main}), stating that we can restrict the desired rotations to involve only marked vertices.
Finally, in \autoref{sec:wrapping-up} we wrap up the previous results to prove \autoref{thm:main}.

\subsection{Good and bad vertices}\label{sec:basic-observations}

Throughout the paper, we assume that all the considered elimination trees are of a same fixed graph $G$.
For simplicity, we may assume henceforth that the considered input graph $G$ is connected.

Our algorithm exploits how a rotation in an elimination tree $T$ may affect the parents and the children of its vertices. Note that a single rotation of an edge $uv \in E(T)$, yielding an elimination tree $T'$, may change the parent of arbitrarily many vertices. Indeed, these vertices are the roots of the red subtrees in \autoref{fig:rotation1}, and the considered vertex $v$ may be adjacent to the root of arbitrarily many subtrees containing at least one vertex adjacent to $u$: for each such root $r$, $\parent(T,r) = v$ but $\parent(T',r) = u$. As a concrete example, in \autoref{fig:eliminationtree}, $\parent(T,s) = v$ but $\parent(T',s) = u$. On the other hand,
item~6 of \autoref{def:rotation} implies that there are at most three vertices whose children set changes from $T$ to $T'$, namely $u,v,z$ as depicted in \autoref{fig:rotation1}. (Note that the sets of children of $u$ and $v$ always change, and that of $z$ changes provided that this vertex exists.) We state this observation formally, since it will be extensively used afterwards.



\begin{observation}\label{obs:ChangingChildren}
One rotation may change the set of children of at most three vertices.
\end{observation}


The above discussion motivates the following definition.

\begin{definition}[bad vertices]\label{def:bad}
Given two elimination trees $T$ and $T'$,
a vertex $v \in V(T)$ is \emph{$(T,T')$-children-bad} (resp. (T,T')-\emph{parent-bad}) if $\child(T,v) \neq \child(T',v)$ (resp. $\parent(T,v) \neq \parent(T',v)$). A vertex $v \in V(T)$ is \emph{$(T,T')$-bad} if it is $(T,T')$-children-bad, or $(T,T')$-parent-bad, or both. A vertex $v \in V(T)$ is \emph{$(T,T')$-good} if it is not $(T,T')$-bad.
\end{definition}

Note that $T$ contains no $(T,T')$-children-bad (or $(T,T')$-parent-bad, or just $(T,T')$-bad) vertices if and only if $T=T'$, that is, if and only if $\dist(T,T')=0$. Also, note that a vertex $v \in V(T)$ is $(T,T')$-children-bad, with $\child(T,v)\neq \emptyset$, if and only if at least one of its children is $(T,T')$-parent-bad. \autoref{obs:ChangingChildren} directly implies the following necessary condition for the existence of a solution.

\begin{observation}\label{prop:+3kno}
Given two elimination trees $T$ and $T'$, if $\dist(T,T') \leq k$ then the number of $(T,T')$-children-bad vertices is at most $3k$.
\end{observation}

\autoref{prop:+3kno} is equivalent to saying that we can safely conclude that any instance $(G,T,T',k)$ of \kelimination with at least $3k+1$ $(T,T')$-children-bad vertices is a {\sf no}-instance. Thus, we can assume henceforth that we are dealing with an instance of \kelimination containing at most $3k$ $(T,T')$-children-bad vertices.



\subsection{Restricting the rotations to small balls around bad vertices}
\label{sec:bounded-diameter}

Our next goal is to prove (\autoref{lemma:restriction-to-balls}) that we can assume that the desired sequence of at most $k$ rotations to transform $T$ into $T'$ uses only edges whose both endvertices lie in the union of all the balls of appropriate radius (depending only on $k$) around $(T,T')$-children-bad vertices of $T$, whose number is bounded by a function of $k$ by \autoref{prop:+3kno}.




In the next definition, for the sake of notational simplicity we omit $T,T'$, and $k$ from the notation $\Bcb$, as we assume that they are already given, and fixed, as the input of our problem. We include $\root(T)$ in the considered set for technical reasons, namely in the proof of \autoref{claim:no-bad}.

\begin{definition}[union of balls of children-bad vertices]\label{def:ball-around-children-bad}
Let $C \subseteq V(T)$ be the set of $(T,T')$-children-bad vertices.
We define $\Bcb = N_T^{2k}[C \cup \root(T)]$.
\end{definition}




\begin{lemma}\label{lemma:restriction-to-balls}
If $\dist(T,T')\leq k$, then there exists an $\ell$-rotation sequence from $T$ to $T'$, with $\ell \leq k$, using only vertices in $\Bcb$.
\end{lemma}
\begin{proof}
Let $\sigma$  be an $\ell$-rotation sequence from $T$ to $T'$, with $\ell \leq k$. Let us denote by ${\sf out}(\sigma)$ the number of edges in $\sigma$ with at least one endvertex not in $\Bcb$. Assuming that ${\sf out}(\sigma) \geq 1$, we proceed to construct another $\ell'$-rotation sequence $\sigma'$ from $T$ to $T'$, with $\ell' \leq \ell$, such that ${\sf out}(\sigma') < {\sf out}(\sigma)$. Repeating this procedure eventually yields a sequence as claimed in the statement of the lemma.

For $i \in [\ell]$, let $u_iv_i$ be the $i$-th edge of $\sigma$ and let $T_i$  be the elimination tree obtained after the rotation of $u_iv_i$. Let also $T_0=T$. For $i \in [\ell]$, we say that a vertex $w \in V(T)$ is \textit{affected} by the rotation of $u_iv_i$ if $\child(T_{i-1},w) \neq \child(T_{i},w)$, and it is \textit{$\sigma$-affected} if it is affected by the rotation of some edge in $\sigma$. 
Recall that a rotation affects at most three vertices, and that these vertices are within distance at most two in the original tree (cf.~vertices $u,v,z$ in \autoref{fig:rotation1}). Moreover, any rotation may increase or decrease vertex distances by at most one, that is, for any $i \in [\ell]$ and any two vertices $x,y \in V(T)$, it holds that
\begin{equation}\label{eq:distance-changes-by-one}
|\dist_{T_{i-1}}(x,y) - \dist_{T_{i}}(x,y)| \leq 1.
\end{equation}

Let $A \subseteq V(T)$ be the set of $\sigma$-affected vertices, and note that $|A| \leq 3k$. The observation above about the fact that the (two or three) vertices affected by a rotation are within distance at most two (in the tree where the rotation is done), together with \autoref{eq:distance-changes-by-one}, imply that for every $Z \in {\sf cc}(T[A])$, it holds that
\begin{equation}\label{eq:Z-bounded-diameter}
\diam(Z) \leq 2k.
\end{equation}
Since by assumption ${\sf out}(\sigma) \geq 1$, there exists $j \in [\ell]$ such that $u_j  \notin \Bcb$ or $v_j  \notin \Bcb$ (or both); assume without loss of generality that $u_j  \notin \Bcb$. Let $Z^{\star} \in {\sf cc}(T[A])$ be the connected component of $T[A]$ containing vertices $u_j$ and $v_j$ (note that they indeed lie in the same component of $T[A]$ since edge $u_jv_j$ belongs to $\sigma$). Let $C \subseteq V(T)$ be the set of $(T,T')$-children-bad vertices, and recall that $\Bcb=N_T^{2k}[C \cup \root(T)]$. See \autoref{fig:bounded-diam}.

\begin{figure}[h!tb]
    \centering
    \vspace{-.35cm}
    \includegraphics[width=.5\textwidth]{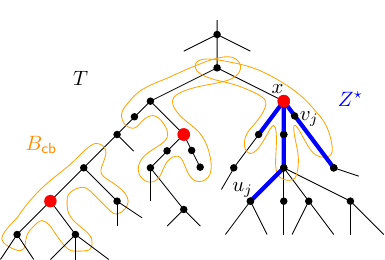}
    \caption{Illustration of the proof of  \autoref{lemma:restriction-to-balls}. $(T,T')$-children-bad vertices are depicted in red, and the balls of radius $2k$ around them are shown with orange bubbles. The connected component $Z^{\star} \in {\sf cc}(T[A])$ containing both vertices $u_j$ and $v_j$  is depicted with thick blue edges. Distances in the figure are not meant to be accurate, and an extremity of an edge without a vertex means that $T$ continues in that direction.\label{fig:bounded-diam}}
\end{figure}

\begin{claim}\label{claim:no-children-bad}
    No vertex in $Z^{\star}$ is $(T,T')$-children-bad.
\end{claim}
\begin{cproof}
Suppose towards a contradiction that the statement does not hold, and let $x \in Z^{\star} \cap C$. Since $u_j  \notin \Bcb$ (see~\autoref{fig:bounded-diam}), the definition of $\Bcb$ implies that $\dist_T(x,u_j) \geq 2k+1$, contradicting \autoref{eq:Z-bounded-diameter} because both $x$ and $u_j$ belong to $Z^{\star}$.
\end{cproof}

\begin{claim}\label{claim:no-bad}
    All vertices in  $Z^{\star}$ are $(T,T')$-good.
\end{claim}
\begin{cproof}
    By \autoref{claim:no-children-bad}, we only need to prove that
    no vertex in $Z^{\star}$ is $(T,T')$-parent-bad. Since $T[Z^{\star}]$ and no vertex in $Z^{\star}$ is $(T,T')$-children-bad, the only vertex in $Z^{\star}$ that may be $(T,T')$-parent-bad is its root, say $x$. Vertex $x$ cannot be the root of $T$, since in that case, the definition of $\Bcb$ and \autoref{eq:Z-bounded-diameter} would imply that $u_j \in \Bcb$, a contradiction. Thus, since $x$ is not the root of $T$, it has a parent $y$ in $T$. But then, if $x$ were $(T,T')$-parent-bad, then $y$ would be $(T,T')$-children-bad, so $y \in A$, implying in turn that $y$ would also belong to the connected component $Z^{\star}$ of $T[A]$, contradicting the fact that $\root(T[Z^{\star}])=x$.
\end{cproof}

Relying on \autoref{claim:no-bad}, we define from $\sigma$ an $\ell'$-rotation sequence $\sigma'$ from $T$ to $T'$, with $\ell' \leq \ell$, as follows: $\sigma'$ consists of the (ordered) edges appearing in $\sigma$, except from those with both endvertices lying in the connected component $Z^{\star}$ of $T[A]$.
\begin{claim}\label{claim:indeed-better-sequance}
$\sigma'$ is  an $\ell'$-rotation sequence from $T$ to $T'$ with $\ell' \leq \ell$ and ${\sf out}(\sigma') < {\sf out}(\sigma)$.
\end{claim}
\begin{cproof}
Note first that \autoref{eq:distance-changes-by-one} implies that both endpoints of any edge occurring in $\sigma$ belong to the same connected component of $T[A]$, and therefore removing from $\sigma$ those rotations with both endvertices in $Z^{\star}$ indeed results in a valid $\ell'$-rotation sequence from $T$ to another elimination tree $\hat{T}$ of $G$, in the sense that the edge rotations appearing in $\sigma'$ can indeed be done in a sequential way. Since at least edge $u_jv_j$ has been removed from $\sigma$, it follows that $\ell' < \ell$ (even if $\ell' \leq \ell$ would be enough for our purposes) and that ${\sf out}(\sigma') < {\sf out}(\sigma)$.

To conclude the proof, it just remains to verify that $\hat{T}=T'$. We will do that by verifying that, for every vertex $v \in V(T)$, it holds that $\child(\hat{T},v)=\child(T',v)$ and $\parent(\hat{T},v)=\parent(T',v)$. We distinguish two cases.

Consider first a vertex $v \in V(T) \setminus Z^{\star}$. In this case, $v$ and its neighbors are involved in the same rotations in $\sigma$ and in $\sigma'$. Since $\sigma$ is an $\ell$-rotation sequence from $T$ to $T'$, we get that indeed $\child(\hat{T},v)=\child(T',v)$ and $\parent(\hat{T},v)=\parent(T',v)$.

Finally, consider  a vertex $v \in Z^{\star}$. By \autoref{claim:no-bad}, $v$ is $(T,T')$-good, implying that $\child(T,v)=\child(T',v)$ and $\parent(T,v)=\parent(T',v)$. Since no rotation of $\sigma'$ involves $v$, we get that $\child(\hat{T},v)=\child(T',v)$ and $\parent(\hat{T},v)=\parent(T',v)$.
\end{cproof}
The above claim concludes the proof of the lemma.\end{proof}

By \autoref{lemma:restriction-to-balls}, we focus henceforth on trying to find an $\ell$-rotation sequence from $T$ to $T'$, with $\ell \leq k$, consisting only of edges with both endvertices in $\Bcb$. First, we will consider each of the at most $3k+1$ connected components of $T[\Bcb]$ separately. In fact, we can get a better bound, as if $T[\Bcb]$ has at least $k+1$ connected components, we can immediately conclude that we are dealing with a {\sf no}-instance, since at least one rotation is needed in each component. Thus, we may assume that $T[\Bcb]$ has at most $k$ connected components. On the other hand, since $T[\Bcb]$ is defined as the union of at most $3k+1$ balls of radius $2k$, it follows that every $Z \in {\sf cc}(T[\Bcb])$ satisfies
\begin{equation}\label{eq:bounded-diameter}
\diam(Z) \leq (3k+1)4k.
\end{equation}

Thus, by \autoref{eq:bounded-diameter}, the ``only'' obstacle to obtain the desired \FPT algorithm is that the vertices in $\Bcb$ can have an arbitrarily large degree. Note that in the particular case where the underlying graph $G$ has bounded degree, for instance if $G$ is a path~\cite{cleary2009rotation,kanj2017computing,li20233,lubiw2015flip}, the maximum degree of any elimination tree of $G$ is bounded, and therefore in that case $|\Bcb|$ is bounded by a function of $k$, and  an \FPT algorithm follows immediately. To the best of our knowledge, this result was not known for graphs other than paths.

\subsection{Description of the marking algorithm}
\label{sec:description-marking}

As discussed in \autoref{sec:sketch}, our strategy to deal with high-degree vertices in $\Bcb$ is as follows. For each connected component $Z \in {\sf cc}(T[\Bcb])$, our goal is to identify a subset $M_Z \subseteq V(Z)$ of size bounded by a function of $k$, such that we can restrict our search to rotations involving only pairs of vertices in $M_Z$. Clearly, this would yield the desired \FPT algorithm. To find such a ``small'' set $M_Z \subseteq V(Z)$, we define the notion of \textit{type} of a vertex $v \in V(Z)$, in such a way that the number of different types is bounded by a function of $k$. Then, we will prove that it is enough to keep in $M_Z$, for each type, a number of vertices bounded again by a function of $k$. 

Let henceforth $Z$ be a connected component of $T[\Bcb]$, which we consider as a rooted tree with its own set of leaves, which are not necessarily leaves in $T$. We define $\root(Z)$ to be the vertex in $V(Z)$ closest to $\root(T)$ in $T$.

Before defining the types, we need to define the \textit{trace} of a vertex $v$ in a designated vertex set $Z \subseteq V(T)$ that will correspond to a connected component of $\Bcb$. Roughly speaking, the trace of a vertex $v$ captures ``abstractly'' the neighborhood of a (whole) subtree rooted at $v$ among (the ordered set of) its ancestors within the designated vertex set $Z \subseteq V(T)$. We stress that, when considering the neighborhood in the set of ancestors,  we look at the whole subtree $T(v)$ rooted at $v$, and not only at its restriction to the set $Z$.


\begin{definition}[trace of a vertex in a component $Z$]
\label{def:trace}
Let $T$ be an elimination tree (of a graph $G$), let $Z$ be a rooted subtree of $T$ corresponding to a connected component of $\Bcb$, and let $v \in V(Z)$. The \emph{trace of $v$ in $Z$}, denoted by $\trace(T,Z,v)$, is a binary vector of dimension $\dist_T(v,\root(Z))$ defined as follows (note that if $v=\root(Z)$, then its trace is empty). For $i \in [\dist_T(v,\root(Z))]$, let $u_i \in \anc(T,v)$ be the ancestor of $v$ in $T$ such that $\dist_T(v,u_i)=i$. Then the $i$-th coordinate of $\trace(T,Z,v)$ is $1$ if $wu_i \in E(G)$ for some vertex $w \in V(T(v))$, and $0$ otherwise.
\end{definition}

See \autoref{fig:example-trace} for an example of the trace of some vertices in a component $Z$.

\begin{figure}[h!tb]
    \centering
    \includegraphics[width=.75\textwidth]{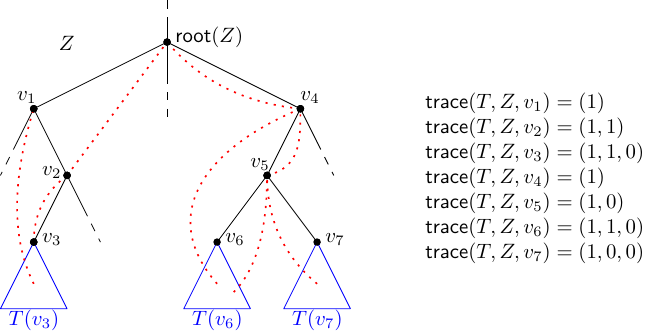}
    \caption{A component $Z$ of $T[\Bcb]$ and the trace of some of its vertices $v_1, \ldots,v_7$. Red dotted edges represent adjacencies in $G$. Note the  ${\sf trace}(T,Z,v_3) = {\sf trace}(T,Z,v_6)$, even if $v_3$ and $v_6$ are not siblings in $Z$.\label{fig:example-trace}}
\end{figure}


For a vertex $v \in V(T)$, let $\wantparent(T,T',v)$ be equal to $\emptyset$ if $\parent(T,v) = \parent(T',v)$, and to $\parent(T',v)$ otherwise. Note that, by \autoref{prop:+3kno}, the function $\wantparent(T,T',v)$ can take up to $3k+1$ distinct values when ranging over all $v \in V(T)$.



\begin{definition}[type of a vertex in a component $Z$]
\label{def:type}
Let $T$ be an elimination tree (of a graph $G$), let $Z$ be a rooted subtree of $T$ corresponding to a connected component of $\Bcb$, and let $v \in V(Z)$. The \emph{type} of  vertex $v$, denoted by $\tau(T,Z,v)$, is recursively defined as follows, where $\typechild(T,Z,v):=\{\tau(T,Z,u) \mid u \in \child(Z,v)\}$ is the set of types occurring in the children of $v$:
\begin{itemize}
    \item If $v$ is a leaf of $Z$, then $\tau(T,Z,v)$ consists of the pair $(\wantparent(T,T',v),\trace(T,Z,v))$.
    \item Otherwise, $\tau(T,Z,v)$ consists of a tuple $(\wantparent(T,T',v),\trace(T,Z,v), f_v)$, where $f_v: \typechild(T,Z,v) \to [k+1]$ is a mapping defined such that, for every $\tau \in \typechild(T,Z,v)$,
    \begin{equation}\label{eq:definition-type}
      f_v(\tau)= \min\{k+1 \ ,\ |\{u \in \child(Z,v) \mid \tau(T,Z,u) = \tau\}|\}.
    \end{equation}
\end{itemize}
\end{definition}

See \autoref{fig:example-types} for an example for $k=2$ of how the types are computed in a component $Z$.

\begin{figure}[h!tb]
    \centering
    \includegraphics[width=1.00\textwidth]{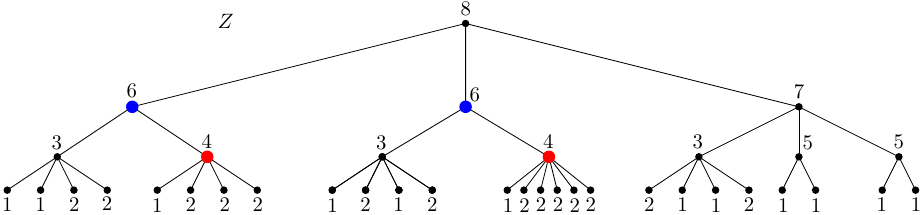}
    \caption{A component $Z$ of $T[\Bcb]$ and the types of its vertices, for an instance with $k=2$. For the sake of simplicity, different types are depicted with different numbers. Assume that the leaves have only two possible types, namely $1$ and $2$, and that all non-leaf vertices at the same distance from the root have the same trace and the same function $\wantparent(T,T',\cdot)$. Note that the red vertices have the same type (namely, $4$) because they both have one child of type $1$ and at least $k+1=3$ children of type $2$. Note also that the blue vertices have the same type (namely, $6$) because they both have one child of type $3$ and one child of type $4$.\label{fig:example-types}}
\end{figure}

\begin{lemma}\label{lem:number-types}
The set $\{\tau(T,Z,v) \mid v \in V(Z)\}$ has size bounded by a function $g(k)$, with
\begin{equation}\label{eq:number-types}
    g(k) = k^{2^{2^{.^{.^{.{2^{\Ocal(k^2)}}}}}}}, \text{ where the tower of exponentials has height $\diam(Z) = \Ocal(k^2)$.}
\end{equation}
\end{lemma}
\begin{proof}
Let $d = \diam(Z)$, and note that $d \leq (3k+1)4k = \Ocal(k^2)$ by \autoref{eq:bounded-diameter}. For $i \in [d]$, let $\tau_i$ be the number of distinct types among the vertices in $V(Z)$ that are at distance exactly $i$ from  $\root(Z)$ in $T$. Formally,
\begin{equation}
\tau_i = |\{\tau(T,Z,v) \mid \dist_T(v,\root(Z)) = i\}|.
\end{equation}
By the definition of type, $\tau_d \leq 2^d \cdot (3k+1)$ since, on the one hand, all vertices at distance $d$ from $\root(Z)$ are leaves of $Z$, and the number of possible traces among leaves is at most $2^d$, and on the other hand the term $3k+1$ corresponds to the possible distinct values of the function $\wantparent(T,T',v)$. For every $i \in [d-1]$, \autoref{eq:definition-type} implies that
\begin{equation}\label{eq:computation-types}
    \tau_i \leq (3k+1) \cdot 2^i \cdot \sum_{j=i+1}^d (k+2)^{\tau_{j}},
\end{equation}
where the term $3k+1$ again comes from the possible distinct values of the function $\wantparent(T,T',v)$, the term $2^i$ comes from the possible different traces within distance $i$ from $\root(Z)$, and the term $k+2$ follows from the fact that, for every $v \in V(Z)$, the function $f_v$ can take up to $k+1$ values for each type $\tau$ of a children of $v$, together with the possibility that a type is not present among the children of $v$. Note that a vertex $v \in V(T)$ with $\dist(v,\root(Z)) = i$ may have children being roots of any possible subtree with diameter at most $d-i$, justifying the sum in  \autoref{eq:computation-types}. Clearly, the upper bound of \autoref{eq:computation-types} is maximized for $i=1$, that is, for the children of $\root(Z)$, yielding the bound claimed in \autoref{eq:number-types}.
\end{proof}

Note that, in order to compute the type of a vertex in a component $Z$, the recursive definition of the types together with \autoref{lem:number-types} easily imply the following observation, where the term $|V(G)|$ comes from checking the neighborhood of $T(v)$ within the set $\anc(T,v)$ in the computation of the trace (cf. \autoref{def:trace}).

\begin{observation}\label{obs:computation-types}
Let $T$ be an elimination tree of a graph $G$, let $Z$ be a rooted subtree of $T$ corresponding to a connected component of $\Bcb$, and let $v \in V(Z)$. Then  $\tau(T,Z,v)$ can be computed in time $g(k) \cdot |V(G)|$, where $g(k)$ is the function from \autoref{lem:number-types}.
\end{observation}

We will now use the notion of type and the bound given by \autoref{lem:number-types} to define the desired set
$M_Z \subseteq Z$ of size bounded by a function of $k$.
In order to find $M_Z$, we apply a  marking algorithm on $Z$, that first identifies a set $M_Z^{\sf pre} \subseteq V(Z)$ of \textit{pre-marked} vertices, whose size is not necessarily bounded by a function of $k$, and then ``prunes'' this set $M_Z^{\sf pre}$ in a root-to-leaf fashion to find the desired  set of \textit{marked} vertices $M_Z \subseteq M_Z^{\sf pre}$ of appropriate size.


Start with $M_Z^{\sf pre}= \emptyset$. For every vertex $v \in V(Z)$ and every $\tau \in \typechild(Z,v)$, do the following:
\begin{itemize}
  \item If $|\{u \in \child(Z,v) \mid \tau(Z,u) = \tau\}| \leq k +1$, add the whole set $\{u \in \child(Z,v) \mid \tau(Z,u) = \tau\}$ to $M_Z^{\sf pre}$.
  \item Otherwise, add to $M_Z^{\sf pre}$ an arbitrarily chosen subset of $\{u \in \child(Z,v) \mid \tau(Z,u) = \tau\}$ of size $k+1$.
\end{itemize}

Finally, add $\root(Z)$ to $M_Z^{\sf pre}$. We define $M_{\sf pre} = \cup_{Z \in {\sf cc}(T[\Bcb])}M_Z^{\sf pre}$ and we call it the set of \emph{pre-marked vertices} of $T$.

We are now ready to define our bounded-size set $M_Z \subseteq M_Z^{\sf pre}$. Start with $M_Z=\{\root(Z)\}$ and for $i = 0,\ldots, \diam(Z)-1$, proceed inductively as follows: if $v \in V(Z)$ is a vertex with $\dist_Z(v,\root(Z))=i$ that already belongs to $M_Z$, add to $M_Z$ the set $\child(Z,v) \cap M_Z^{\sf pre}$. Finally, for every $(T,T')$-children-bad vertex $v$ of $T$ that belongs to $Z$, we add to $M_Z$ the set $\anc(Z,v)$. This concludes the construction of $M_Z$. Note that if a vertex $v \in V(Z)$ belongs to $M_Z$, then the whole set $\anc(Z,v)$ belongs to $M_Z$ as well. We define $M = \cup_{Z \in {\sf cc}(T[\Bcb])}M_Z$, and we call it the set of \emph{marked vertices} of $T$. See \autoref{fig:example-marking} for an example of the marking algorithm.

\begin{figure}[h!tb]
    \centering
    \vspace{-.25cm}
    \hspace{-.75cm}\includegraphics[width=1.05\textwidth]{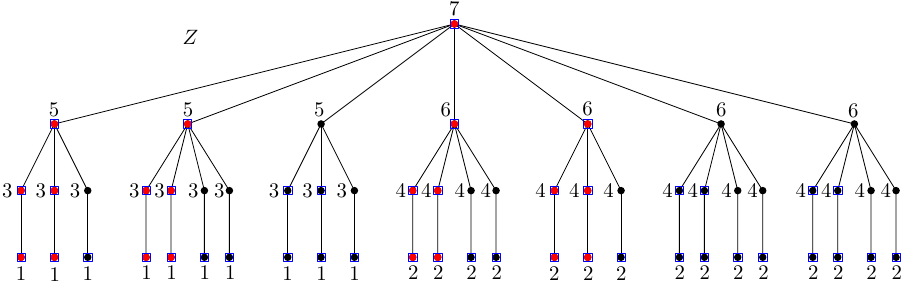}
    \caption{Example of the marking algorithm applied to a component $Z$ of $T[\Bcb]$, for an instance with $k=1$. As in \autoref{fig:example-types}, different types are depicted with different numbers. Vertices inside blue squares belong to $M_Z^{\sf pre}$, and red vertices belong to $M_Z$.\label{fig:example-marking}}
\end{figure}


\begin{lemma}\label{lem:M-bounded}
The set $M \subseteq V(T)$ of marked vertices has size bounded by a function $h(k)$,  where $h(k)$ has the same asymptotic growth as the function $g(k)$ given by \autoref{lem:number-types}. Moreover, $M$ can be computed in time $h(k) \cdot |V(G)|$.
\end{lemma}
\begin{proof}
  Let us analyze the size of each set $M_Z$ of $T[M]$ separately, as their number is at most $3k+1$, and this factor gets subsumed by the asymptotic growth of $h(k)$. The diameter of $T[M_Z]$ is at most $(3k+1)4k = \Ocal(k^2)$ by \autoref{eq:bounded-diameter}, so in order to bound the size of $M_Z$, it just remains to bound the degree of $T[M_Z]$. By construction of $M_Z$, every vertex $v \in M_Z$ has at most $\child(Z,v) \cap M_Z^{\sf pre} + 3k$ children, where the term $3k$ corresponds to the maximum number of $(T,T')$-children-bad vertices (cf. \autoref{prop:+3kno}) that, together with their ancestors, have been marked within $Z$. Since the set $M_Z^{\sf pre}$ is defined by pre-marking, for each vertex, at most $k+1$ children of each type, the set $\child(Z,v) \cap M_Z^{\sf pre}$ has size at most $g(k) \cdot (k+1)$, where $g(k)$ is the function given by \autoref{lem:number-types}. Thus, $|M| \leq h(k)$, where $h(k)$ has the same asymptotic growth as $g(k)$.

  As for computing the set $M$ in time $h(k) \cdot |V(G)|$, it follows from the definition of $M$ (that uses the set of pre-marked vertices $M_Z^{\sf pre}$), that can be computed in time that is asymptotically dominated by computing the type of a vertex, which is bounded by $g(k) \cdot |V(G)|$ by \autoref{obs:computation-types}.
\end{proof}

\subsection{Restricting the rotations to marked vertices}
\label{sec:restriction-to-marked}

In this subsection we prove our main technical result (\autoref{lem:main}), which immediately yields the desired \FPT algorithm combined  with  \autoref{lem:M-bounded} (whose proof uses  \autoref{lemma:restriction-to-balls}), as discussed in \autoref{sec:wrapping-up}. We first need an easy lemma that will be extensively used in the proof of \autoref{lem:main}.

\begin{lemma}\label{lem:few-siblings-involved}
Let $\sigma$ be an $\ell$-rotation sequence from $T$ to $T'$, for some $\ell \leq k$. For every vertex $v \in V(T)$, there are at most  $k$ vertices  $u_1, \ldots, u_k \in \child(T,v)$ such that $\sigma$ uses a vertex in each of the rooted subtrees $T(u_1), \ldots, T(u_k)$.
\end{lemma}
\begin{proof}
Assume towards a contradiction that there is a vertex $v \in V(T)$ having   $k+1$ children, say $u_1, \ldots u_{k+1}$, such that $\sigma$ uses $k+1$ vertices $u_1', \ldots, u_{k+1}'$ with $u_r' \in T(u_r)$ for $r \in [k+1]$. Then, since $\sigma=(e_1,\ldots,e_{\ell})$ is made of at most $k$ rotations, by the pigeonhole principle necessarily there exist two children $u_i,u_j$ of $v$ and an integer $p \in [\ell]$ such that $e_p = u'_iu'_j$, and none of $u'_i,u'_j$ occurs in any other rotation of $\sigma$ other than $e_p$. Since $e_p = u'_iu'_j$, it means that in the elimination tree where this rotation takes place, namely $T_{p-1}$, it holds that $u'_iu'_j \in E(T_{p-1})$. See \autoref{fig:important-easy-lemma} for an illustration.

On the other hand, note that if $T_2$ is an elimination tree resulting from an elimination tree $T_1$ after the rotation of an edge $wz \in E(T_1)$, and $a,b$ are vertices of $T_1$ (and $T_2$) such that $ab \notin E(T_1)$ and $ab \in E(T_2)$, then necessarily $\{a,b\} \cap \{w,z\} \neq \emptyset$, that is, necessarily the rotation involves at least one of $a$ and $b$. See \autoref{fig:rotation1} for a visualization of this claim, where the new edges that appear after the rotation of $uv$ are $zv$ and the edges between $u$ and some of the red subtrees: each of these new edges contains $u$ or $v$.

Since $u'_iu'_j \notin E(T_0) = E(T)$ because $u_i' \in T(u_i), u_j' \in T(u_j)$, and $u_i,u_j$ are $T$-siblings, and $u'_iu'_j \in E(T_{p-1})$, by the above paragraph there exists some integer $q \in [\ell]$, with $q < p$, such that the rotation $e_q$ of $\sigma$ contains at least one of $u'_i$ and $u'_j$. This contradicts that fact that none of $u'_i,u'_j$ occurs in any other rotation of $\sigma$ other than $e_p$.
\end{proof}

\begin{figure}[h!tb]
    \centering
    \vspace{-.3cm}
    \includegraphics[width=.3\textwidth]{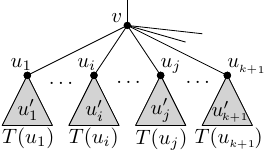}
    \caption{Illustration of the proof of \autoref{lem:few-siblings-involved}.\label{fig:important-easy-lemma}}
\end{figure}

Note that, if in the statement of \autoref{lem:few-siblings-involved} we replaced ``at most $k$ vertices'' with ``at most $2k$ vertices'', then its proof would be trivial, as any of the at most $k$ rotations of $\sigma$ involves two vertices, so at most $2k$ distinct vertices overall. In that case, for the proof of \autoref{lem:main} to go through,  we would have to replace, in \autoref{eq:definition-type} in the definition of type, ``$k+1$'' with ``$2k+1$'' when taking the minimum. In the sequel we will often use a weaker version of \autoref{lem:few-siblings-involved}, namely that for every vertex $v \in V(T)$, at most $k$ vertices in $\child(T,v)$ are used by an $\ell$-rotation sequence from $T$ to $T'$.


We are now ready to prove our main lemma.

\begin{lemma}\label{lem:main}
If $\dist(T,T')\leq k$, then there exists an $\ell$-rotation sequence from $T$ to $T'$, with $\ell \leq k$, using only vertices in $M$.
\end{lemma}
\begin{proof}
Let $\sigma$ be an $\ell$-rotation sequence from $T$ to $T'$, for some $\ell \leq k$, minimizing, among all $\ell$-rotation sequences from $T$ to $T'$, the number of vertices in $V(T) \setminus M$ (that is, the non-marked vertices) used by $\sigma$. Note that $\sigma$ exists by the hypothesis that $\dist(T,T')\leq k$. If there are no vertices in $V(T) \setminus M$ used by $\sigma$, then we are done, so assume that there are. Our goal is to define another $\ell$-rotation sequence $\sigma'$ from $T$ to $T'$ using strictly less vertices in $V(T) \setminus M$ than $\sigma$, contradicting the choice of $\sigma$ and concluding the proof.

To this end, let $v \in V(T) \setminus M$ be a furthest (with respect to the distance to $\root(T)$) non-marked vertex of $T$ that is used by $\sigma$. By \autoref{lemma:restriction-to-balls}, we can assume that $v \in \Bcb$. Let $Z$ be the connected component of $T[\Bcb]$ such that $v \in Z$. For technical reasons, it will be helpful to assume that $v$ is not a leaf of $Z$. (This can be achieved, for instance, by observing that the analysis of the size of the components $Z$ in \autoref{lemma:restriction-to-balls} is not tight. Alternatively, we can just ``artificially''  increase their diameter by one~--i.e., replacing $2k$ with $2k+1$ in the definition of $\Bcb$--~so that we can safely assume that the leaves of $Z$ are never used by a rotation sequence.) Note that the choice of $v$ as a lowest (i.e., furthest) non-marked vertex of $T$ used by $\sigma$ implies that for every vertex $u \in \child(T,v)$,  the whole subtree $T(u)$ remains intact throughout $\sigma$, meaning that it appears as a rooted subtree in all the intermediate elimination trees generated by the $\ell$-rotation sequence $\sigma$. We distinguish two cases, the second one being considerably more involved, but that will benefit from the intuition developed in the first one.

\bigskip
\noindent\fbox{\textbf{Case 1}: $v$ has a marked $T$-sibling $v'$ with $\tau(T,Z,v)=\tau(T,Z,v')$.}
\medskip

Since $v$ is non-marked and by assumption it has some marked $T$-sibling, the definition of $M$ (namely, that up to $k+1$ vertices of each type are recursively marked) and \autoref{lem:few-siblings-involved} imply that $v$ has some marked $T$-sibling of the same type that is {\sl not} used by~$\sigma$. Let without loss of generality $v'$ be such a $T$-sibling of $v$. Note that the choice of $v$ as a lowest non-marked vertex used by $\sigma$ implies that the whole subtree $T(v')$ remains intact throughout $\sigma$. See \autoref{fig:mainlemma-case1} for an illustration. In this case, we define $\sigma'$ from $\sigma$ by just replacing $v$ with $v'$ in all the rotations of $\sigma$ involving $v$. We need to prove that $\sigma'$ is well-defined (that is, that the edges to be rotated do exist in the intermediate elimination trees) and that it is an $\ell$-rotation sequence from $T$ to $T'$. Once this is proved, this case is done, as $\sigma'$ uses strictly more marked vertices than $\sigma$.

\begin{claim}\label{claim:case1-well-defined}
$\sigma'$ is a well-defined $\ell$-rotation sequence.
\end{claim}
\begin{cproof}
  We need to prove that if $\sigma=(e_1,\ldots,e_{\ell})$ and $e_i=vw$ for some $i \in [\ell]$ and some $w \in V(T)$, then $v'w \in E(T'_{i-1})$.
  Let us prove it by induction. For $i=1$, assume that $e_1=vw$ for some $w \in V(T)$. The choice of $v$ as a lowest non-marked vertex of $T$ used by $\sigma$ implies that $w = \parent(T,v)$, and since $v'$ is a $T$-sibling of $v$, it follows that $v'w \in E(T'_{0})=E(T)$.
  Assume now inductively that the edges to be rotated exist up to $i-1$,
   and suppose that $e_{i+1}=vw$ for some $w \in V(T)$. Since $\tau(T,Z,v)=\tau(T,Z,v')$, it follows in particular that $\trace(T,Z,v)=\trace(T,Z,v')$ (see the green dotted edges in \autoref{fig:mainlemma-case1}), and as $v$ and $v'$ are $T$-siblings, it follows that the rooted subtrees $T(v)$ and $T(v')$ have exactly the same neighbors in the set $V(Z) \setminus (V(T(v)) \cup V(T(v')))$, that is,
  \begin{equation}\label{eq:case1-same-neighbors}
    N_G(T(v)) \cap  (V(Z) \setminus (V(T(v)) \cup V(T(v')))= N_G(T(v'))) \cap  (V(Z) \setminus (V(T(v)) \cup V(T(v'))).
  \end{equation}

\autoref{eq:case1-same-neighbors} implies that, up to $i-1$, vertex $v'$ has been following exactly the same moves in $\sigma'$ as the ones followed by vertex $v$ in $\sigma$. Thus, the fact that $vw \in E(T_i)$ implies that $v'w \in E(T'_i)$, and the claim follows.
\end{cproof}

\begin{figure}[h!tb]
    \centering
    \vspace{-.3cm}
    \includegraphics[width=.65\textwidth]{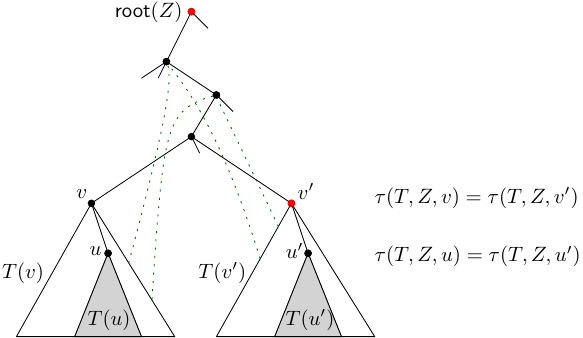}
    \caption{Illustration of Case~1 in the proof of  \autoref{lem:main}. Vertex $v$ is non-marked and used by $\sigma$, and its $T$-sibling $v'$ is marked (in red) and not used by $\sigma$.\label{fig:mainlemma-case1}}
\end{figure}

\begin{claim}\label{claim:case1-good-rotation-sequence}
$\sigma'$ is an $\ell$-rotation sequence from $T$ to $T'$.
\end{claim}
\begin{cproof}
By \autoref{claim:case1-well-defined}, $\sigma'$ is an $\ell$-rotation sequence from $T$ to some elimination tree $\hat{T}$ of $G$. It remains to prove that $\hat{T}=T'$. This is equivalent to proving that, for every vertex $u \in V(G)$ that is not a root in $T'$ or $\hat{T}$, $\parent(T',u) = \parent(\hat{T},u)$. By definition of $\sigma'$, this is clearly the case for every vertex $u$ that is not in the set $\{v\}\cup \{v'\} \cup \child(T,v) \cup \child(T,v')$.

Consider first a vertex $u \in \child(T,v)$. Note that the whole rooted subtree $T(v)$ remains intact throughout $\sigma'$, in the same way as $T(v')$ remains intact throughout $\sigma$. Moreover, since all $(T,T')$-children-bad vertices belong to $M$, and $v$ does not, it follows that $v$ is not $(T,T')$-children-bad, that is, that $\child(T,v)=\child(T',v)$. This implies that $v =  \parent(T,u) = \parent(T',u) = \parent(\hat{T},u)$ for every vertex $u \in \child(T,v)$.

Consider now vertices $v$ and $v'$. The choice of $v$ as a lowest non-marked vertex used by $\sigma$ implies that all the descendants of $v$ in $T$ are $(T,T')$-good, except maybe $v$ itself that may be $(T,T')$-parent-bad. If that is the case, the hypothesis that $\tau(T,Z,v)=\tau(T,Z,v')$ and the fact that the function $\wantparent(T,T',\cdot)$ is part of the definition of the type of a vertex imply that $v'$ is also $(T,T')$-parent-bad and that $\wantparent(T,T',v) = \wantparent(T,T',v')$.  Recall that \autoref{eq:case1-same-neighbors} discussed above implies that
$v'$ follows in $\sigma'$ the same moves that $v$ follows in $\sigma$. And since $\wantparent(T,T',v) = \wantparent(T,T',v')$, including the case where both sets are empty, it follows that $\parent(T',v) = \parent(T',v') = \parent(\hat{T},v')=\parent(\hat{T},v)$.

Finally, consider a vertex $u' \in \child(T,v')$. Note that the whole subtree $T(u')$ remains intact throughout $\sigma'$, in the same way as the whole subtree $T(u)$ remains intact throughout $\sigma$ for every vertex $u \in \child(T,v)$. Since such a vertex $u' \in \child(T,v')$ is $(T,T')$-good by the choice of $v$, we have that $v' = \parent(T,u')= \parent(T',u')$. The fact that $\tau(T,Z,v)=\tau(T,Z,v')$ implies, similarly to the discussion after
\autoref{eq:case1-same-neighbors}, that the subtree $T(u')$ follows in $\sigma'$ the same moves followed by $T(u)$ in $\sigma$, where $u \in \child(T,v)$ is a vertex such that $\tau(T,Z,u)=\tau(T,Z,u')$, which exists by the recursive definition of type (cf.~\autoref{def:type}) and the hypothesis that $\tau(T,Z,v)=\tau(T,Z,v')$; see \autoref{fig:mainlemma-case1}. Thus, $v' = \parent(T,u')= \parent(T',u') = \parent(\hat{T},u')$, and the claim follows.
\end{cproof}

\medskip
\noindent\fbox{\textbf{Case 2}: all $T$-siblings $v'$ of $v$ with $\tau(T,Z,v)=\tau(T,Z,v')$, if any, are non-marked.}
\medskip

In this case, in order to define another $\ell$-rotation sequence $\sigma'$ from $T$ to $T'$ that uses more marked vertices than $\sigma$, we need to modify $\sigma$ in a more global way than what we did in Case~1 above, where it was enough to replace vertex $v$ with a marked $T$-sibling of the same type. Now, in order to define $\sigma'$, we need a more global replacement. To this end, the following claim guarantees the existence of a very helpful vertex $v^{\star}$. See \autoref{fig:mainlemma-case2} for an illustration.

\begin{claim}\label{claim:vertex-vstar-exists}
There exists a unique vertex $v^{\star} \in \anc(Z,v)$ such that
\begin{itemize}
  \item all vertices in $T(v^{\star})$ are non-marked, 
  \item $v^{\star}$  has a marked $T$-sibling~$v'$ such that
      \begin{itemize}
        \item $\tau(T,Z,v^{\star})=\tau(T,Z,v')$, and
        \item no vertex in $T(v')$ is used by $\sigma$, and
      \end{itemize}
  \item $v^{\star}$  is the vertex closest to $v$ satisfying the above properties.
\end{itemize}
\end{claim}
\begin{cproof}
Since $\root(Z) \in M$ and $v \notin M$, the definition of the marking algorithm implies that there exist a vertex $v^{\star} \in \anc(Z,v)$ such that all vertices in $T(v^{\star})$ are non-marked, and a marked $T$-sibling~$v'$ of $v^{\star}$ with $\tau(T,Z,v^{\star})=\tau(T,Z,v')$. Moreover, the fact that $M$ is defined by recursively marking up to $k+1$ vertices of each type implies, together with \autoref{lem:few-siblings-involved} and the fact that the desired vertex $v^{\star}$ is non-marked, imply that we can choose $v'$ such that no vertex in $T(v')$ is used by $\sigma$. Finally, we can choose $v^{\star}$ in a unique way as being the vertex closest to $v$ satisfying the above properties.
\end{cproof}

\begin{figure}[h!tb]
    \centering
    \vspace{-.15cm}
    \hspace{-.75cm}\includegraphics[width=1.05\textwidth]{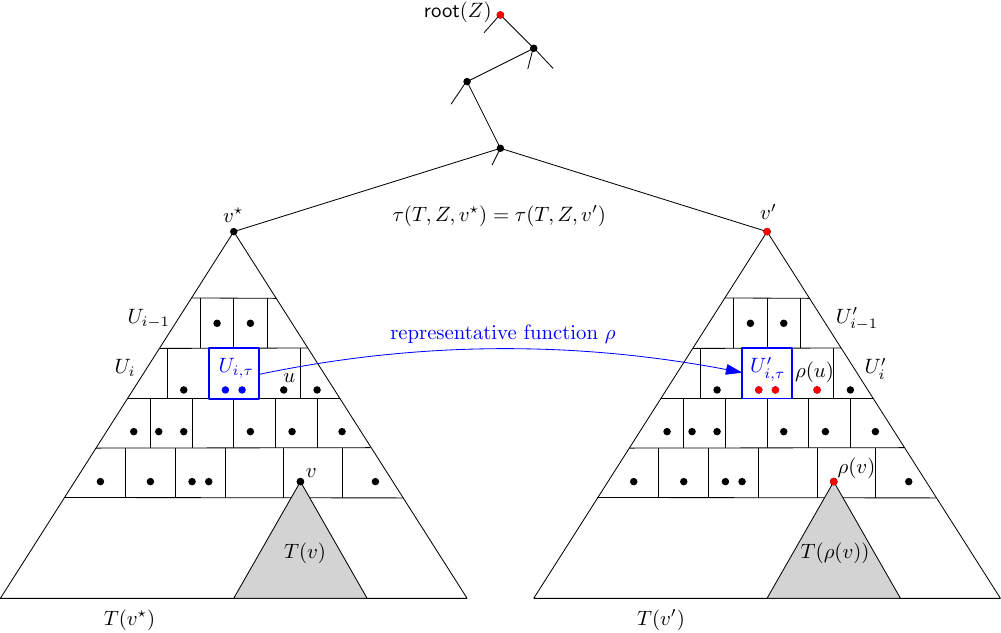}
    \caption{Illustration of Case~2 in the proof of  \autoref{lem:main}. All vertices in $T(v^{\star})$ are non-marked,  and (at least) vertex $v$ is used by $\sigma$. No vertex in $T(v')$ is used by $\sigma$, and (at least) $v'$ and all vertices in the image of $\rho$, including $\rho(v)$, are marked (in red). The sets $U_{i,\tau}$ of vertices used by $\sigma$ in $T(v^{\star}) \cap Z$ are depicted with squares, as well as their images in $T(v') \cap Z$ via the representative function $\rho$.\label{fig:mainlemma-case2}}
\end{figure}

Note that Case~1 of the proof corresponds to the particular case where $v^{\star}$ is equal to $v$ itself, but we prefer to separate both cases for the sake of readability. Intuitively, we will apply recursively the argument of Case~1 to the rooted subtrees $T(v^{\star})$ and $T(v')$, starting with $v^{\star}$ and $v'$, exploiting the definition of types to appropriately define the desired replacement of vertices to construct  $\sigma'$ from  $\sigma$.

Formally, let $U$ be the set of vertices in $V(T(v^{\star})) \cap Z$ used by $\sigma$ (so in Case~1, $U=\{v\}$), and let $\rho: U \to V(T(v')) \cap Z$ be the injective function defined as follows. For $i=0,\ldots,\dist_T(v^{\star},v)$, let $U_i \subseteq U$ be the set of vertices in $V(T(v^{\star})) \cap Z$ used by $\sigma$ that are at distance exactly $i$ from vertex $v^{\star}$ in $T$. Note that some of the sets $U_i$ may be empty, and that $U_{\dist_T(v^{\star},v)}$ contains $v$. For every type $\tau$ occurring in a vertex in $U_i$, let $U_{i,\tau}$ be the set of vertices of type $\tau$ in $U_i$. Note that, if a set $U_i$ is non-empty, then $\{U_{i,\tau} \mid \text{$\tau$ occurs in $U_i$}\}$ defines a partition of $U_i$ into non-empty sets. Let $U'_{i,\tau}$ be a set of marked vertices of type $\tau$ in $V(T(v')) \cap Z$ of size $|U_{i,\tau}|$ (we shall prove in \autoref{claim:representative-function} that it exists). Then we define $\rho|_{U_{i,\tau}}$ as any bijection between $U_{i,\tau}$ and $U'_{i,\tau}$. See \autoref{fig:mainlemma-case2} for an illustration.

\begin{claim}\label{claim:representative-function}
The function $\rho$ is well-defined and injective.
\end{claim}
\begin{cproof}
  Assuming that the sets $U'_{i,\tau}$ exist, it is clear that $\rho$ is injective. Hence, we shall prove that for every $i=0,\ldots,\dist_T(v^{\star},v)$ and every type $\tau$ occurring in a vertex in $U_i$, there exists a set $U'_{i,\tau} \subseteq V(T(v')) \cap Z$ of marked vertices of type $\tau$ with $|U'_{i,\tau}| = |U_{i,\tau}|$. We proceed by induction on $i$. For $i=0$, if $U_0 \neq \emptyset$, then $U_0=\{v^{\star}\}$ and the only type occurring in $U_0$ is $\tau(T,Z,v^{\star})=:\tau$. Thus, we can just take $U'_{0,\tau}=\{v'\}$, where $v'$ is the vertex given by \autoref{claim:vertex-vstar-exists}, so that $\rho(v^{\star}) = v'$. Note that it is not necessarily the case that vertex $v^{\star}$ is used by $\sigma$, which is equivalent, by the above definitions, to $U_0 \neq \emptyset$.

  We now prove the statement for any $i \geq 1$, using the recursive definition of type (cf. \autoref{def:type}). We may assume that $U_{i} \neq \emptyset$, as otherwise there is nothing to prove.
  The facts that $\tau(T,Z,v^{\star})=\tau(T,Z,v')$ and that $v'$ is marked imply that, for every $i \in [\dist_T(v^{\star},v)]$ and every type $\tau$, the following holds:

\begin{itemize}
  \item If $T(v^{\star})$ contains a set $A_{i,\tau}$ of at most $k+1$ vertices of type $\tau$, all at distance exactly $i$ from $v^{\star}$, then $T(v')$ contains a set $A'_{i,\tau}$ of {\sl marked} vertices of type $\tau$ with $|A'_{i,\tau}|=|A_{i,\tau}|$, all at distance exactly $i$ from $v'$.
  \item If $T(v^{\star})$ contains a set $A_{i,\tau}$ of at least $k+2$ vertices of type $\tau$, all at distance exactly $i$ from $v^{\star}$, then $T(v')$ contains a set $A'_{i,\tau}$ of {\sl marked} vertices of type $\tau$ with $|A'_{i,\tau}|=k+1$, all at distance exactly $i$ from $v'$.
\end{itemize}

It is important to pay attention to the difference between the above two items: while in the first one the set $A'_{i,\tau}$ of marked vertices has the same size as $A_{i,\tau}$, in the second one we can ``only'' guarantee, due to \autoref{eq:definition-type}, that $|A'_{i,\tau}|=k+1$, even if $A'_{i,\tau}$ may be arbitrarily larger (even of size not bounded by any function of $k$). Fortunately, thanks to \autoref{lem:few-siblings-involved}, this is enough  for finding the desired set $U'_{i,\tau}$ in order to define the representative function $\rho$, as we proceed to discuss. Note that, for every $i \in [\dist_T(v^{\star},v)]$ and every type $\tau$, it holds that $U_{i,\tau} \subseteq A_{i,\tau}$, since $\sigma$ may use only some of the vertices in  $A_{i,\tau}$.

Suppose first that, for some $i \in [\dist_T(v^{\star},v)]$ and some type $\tau$, the first item above holds. Then, since $U_{i,\tau} \subseteq A_{i,\tau}$, we can just take $U'_{i,\tau}$ as any subset of $A'_{i,\tau}$ of size $|U_{i,\tau}|$.

Suppose now that the second item above holds, that is, that $T(v^{\star})$ contains a set $A_{i,\tau}$ of at least $k+2$ vertices of type $\tau$, all at distance exactly $i$ from $v^{\star}$. By \autoref{lem:few-siblings-involved}, it holds that $U_{i,\tau} \leq k$, and since $|A'_{i,\tau}|=k+1$, we can indeed define $U'_{i,\tau}$ as any subset of $A'_{i,\tau}$ of size $|U_{i,\tau}|$, and the claim follows.
\end{cproof}

For every vertex $u \in U$ (recall that $U$ is the set of vertices in $V(T(v^{\star})) \cap Z$ used by $\sigma$), the vertex $\rho(u)$ is called the \emph{representative} of $u$. Note that the function $\rho$ is also defined on~$v$, since it is used by $\sigma$. We now define $\sigma'$ from $\sigma$ by replacing, in the rotations defining the sequence, every vertex $u \in U$ by its representative~$\rho(u)$.

The following two claims correspond respectively to \autoref{claim:case1-well-defined} and \autoref{claim:case1-good-rotation-sequence} of Case~1, and conclude the proof of the lemma.

\begin{claim}\label{claim:case2-well-defined}
$\sigma'$ is a well-defined $\ell$-rotation sequence.
\end{claim}
\begin{cproof} Similarly to the proof of \autoref{claim:case1-well-defined}, we need to prove that if
$\sigma=(e_1,\ldots,e_{\ell})$, then for every $i \in [\ell]$ the corresponding edge to be rotated  in $\sigma'$ exists in the intermediate subtree.  Suppose that $uw$ is a rotation in $\sigma$, for some $u,w \in V(T)$, such that $uw \in E(T_{i-1})$ for some $i \in [\ell]$. We distinguish three cases depending to whether the vertices $u,w$ involved in the rotation belong to $T(v^{\star})$ or not.

\begin{itemize}
  \item Suppose first that none of $u,w$ belongs to $T(v^{\star})$. It is not difficult to verify that if $uw \in E(T_{i-1})$, then $uw \in E(T'_{i-1})$ as well, where $T'_{i-1}$ is the tree obtained from $T$ by applying the first $i-1$ rotations of $\sigma'$. Indeed, the fact that $\tau(T,Z,v^{\star})=\tau(T,Z,v')$ implies that the existence of such an edge with both endpoints outside $T(v^{\star})$ is preserved when replacing the vertices in $U$ with their representatives.

  \item Suppose now that both $u,w$ belong to $T(v^{\star})$. In this case, the edge $uw$ of $\sigma$ has been replaced by $\rho(u) \rho(w)$ in $\sigma'$. Note that both vertices $\rho(u), \rho(w)$ belong to $T(v')$. The definition of the representative function $\rho$ implies that $\tau(T,Z,u)=\tau(T,Z,\rho(u))$ and $\tau(T,Z,w)=\tau(T,Z,\rho(w))$, which in particular implies that $\trace(T,Z,u)=\trace(T,Z,\rho(u))$ and $\trace(T,Z,w)=\trace(T,Z,\rho(w))$. It follows that
       vertex $\rho(u)$ (resp. $\rho(w)$) has been following the same moves within $T(v')$ in $\sigma'$ as the ones followed by vertex $u$ (resp. $w$) within $T(v^{\star})$ in $\sigma$.
       Thus, the fact that $uw \in E(T_{i-1})$ implies that $\rho(u) \rho(w) \in E(T'_{i-1})$.

  \item Finally, suppose without loss of generality that $u \in V(T(v^{\star}))$ and $w \notin V(T(v^{\star}))$. This case is similar to the proof of \autoref{claim:case1-well-defined}, with the role of $v$ replaced with $v^{\star}$. Namely, it can be proved by induction on $i$ in a similar fashion. For $i=1$, assume that $e_1=uw$ with $u \in V(T(v^{\star}))$ and $w \notin V(T(v^{\star}))$. Then, by the definition of $v^{\star}$ (cf. \autoref{claim:vertex-vstar-exists}), necessarily $u = v^{\star}$ and $w = \parent(T,v^{\star})$. Since $\rho(v^{\star}) = v'$ and  $v'$ is a $T$-sibling of $v^{\star}$, it follows that $\rho(v^{\star})w \in E(T'_{0})=E(T)$.

      Assume now inductively that the edges to be rotated exist up to $i-1$,  and suppose that $e_{i+1}=uw$ for some $u \in V(T(v^{\star}))$ and $w \notin V(T(v^{\star}))$. The fact that $\tau(T,Z,u)=\tau(T,Z,\rho(u))$ implies that $\trace(T,Z,u)=\trace(T,Z,\rho(u))$, which implies in particular that $T(u)$ and $T(\rho(u))$ have exactly the same neighbors in the set $V(Z) \setminus (V(T(v^{\star})) \cup V(T(v')))$, that is,
  \begin{equation}\label{eq:case2-same-neighbors}
    N_G(T(u)) \cap  (V(Z) \setminus (V(T(v^{\star})) \cup V(T(v')))= N_G(T(\rho(u))) \cap  (V(Z) \setminus (V(T(v^{\star})) \cup V(T(v'))).
  \end{equation}

 The fact that $\tau(T,Z,u)=\tau(T,Z,\rho(u))$ and \autoref{eq:case2-same-neighbors} imply that, up to $i-1$, vertex  $\rho(u)$  has been following the same moves within $T(v')$ and $V(Z) \setminus (V(T(v^{\star})) \cup V(T(v')))$ in $\sigma'$ as the ones followed by vertex $u$ within $T(v^{\star})$ and $V(Z) \setminus (V(T(v^{\star})) \cup V(T(v')))$ in $\sigma$.
 Thus, the fact that $uw \in E(T_i)$ implies that $\rho(u)w \in E(T'_i)$, and the claim follows.
\end{itemize}
\vspace{-.05cm}
\end{cproof}

\begin{claim}\label{claim:case2-good-rotation-sequence}
$\sigma'$ is an $\ell$-rotation sequence from $T$ to $T'$.
\end{claim}
\begin{cproof}
The proof of this claim follows that of \autoref{claim:case1-good-rotation-sequence}. By \autoref{claim:case2-well-defined}, $\sigma'$ is an $\ell$-rotation sequence from $T$ to some elimination tree $\hat{T}$ of $G$. It remains to prove that $\hat{T}=T'$. This is equivalent to proving that, for every vertex $u \in V(G)$ that is not a root in $T'$ or $\hat{T}$, $\parent(T',u) = \parent(\hat{T},u)$. By definition of $\sigma'$, this is clearly the case for every vertex $u$ that is not in the set $Z \cap (V(T(v^{\star})) \cup  V(T(v')))$, given that $\trace(T,Z,v^{\star}) = \trace(T,Z,v')$. We distinguish three cases to deal with the vertices in $Z \cap (V( T(v^{\star})) \cup  V(T(v')))$.

\begin{itemize}
\item Consider first a vertex $u \in Z \cap V(T(v^{\star}))$ different from $v^{\star}$. Note that the whole rooted subtree $T(v^{\star})$ remains intact throughout $\sigma'$, in the same way as the whole subtree $T(v')$ remains intact throughout $\sigma$. Moreover, since all $(T,T')$-children-bad vertices belong to $M$, and $v^{\star}$ does not (cf. \autoref{claim:vertex-vstar-exists}), it follows that $v^{\star}$ is not $(T,T')$-children-bad, that is, that $\child(T,u)=\child(T',u)$ for every vertex $u \in Z \cap V(T(v^{\star}))$. This implies that $\parent(T,u) = \parent(T',u) = \parent(\hat{T},u)$ for every vertex $u \in Z \cap V(T(v^{\star}))$ different from~$v^{\star}$.

 \item Consider now vertices $v^{\star}$ and $v'$. The definition of $v^{\star}$ (cf. \autoref{claim:vertex-vstar-exists})
implies that all the descendants of $v^{\star}$ in $T$ are $(T,T')$-good, except maybe $v^{\star}$ itself that may be $(T,T')$-parent-bad. If that is the case, the fact that $\tau(T,Z,v^{\star})=\tau(T,Z,v')$ and the fact that the function $\wantparent(T,T',\cdot)$ is part of the definition of the type of a vertex imply that $v'$ is also $(T,T')$-parent-bad and that $\wantparent(T,T',v^{\star}) = \wantparent(T,T',v')$.  Similarly to \autoref{eq:case1-same-neighbors}, the fact that $\tau(T,Z,v^{\star}) = \tau(T,Z,v')$ implies that
\begin{equation}\label{eq:case2-same-neighbors-again}
    N_G(T(v^{\star})) \cap  (V(Z) \setminus (V(T(v^{\star})) \cup V(T(v')))= N_G(T(v')) \cap  (V(Z) \setminus (V(T(v^{\star})) \cup V(T(v'))),
  \end{equation}
which implies that $v'$ follows in $\sigma'$ the same moves that $v$ follows in $\sigma$. And since $\wantparent(T,T',v) = \wantparent(T,T',v')$, including the case where both sets are empty, it follows that $\parent(T',v^{\star}) = \parent(T',v') = \parent(\hat{T},v')=\parent(\hat{T},v^{\star})$.

\item Finally, consider a vertex $u' \in Z \cap V(T(v')$ different from $v'$. First note that if neither $\parent(T,u')$ nor any descendant of $u$ in $T$ (including $u$ itself) are used by $\sigma'$, then, since no vertex in $T(v')$ is used by $\sigma$, it follows that $\parent(T,u')= \parent(T',u') = \parent(\hat{T},u')$.

    We now proceed recursively, by first considering a lowest vertex $u' \in Z \cap V(T(v')$ such that $\parent(T,u')$ is used by $\sigma'$. Note that such a vertex $u'$ exists because (at least) $\rho(v)$ is used by $\sigma'$ and we may assume that the leaves of $Z \cap V(T(v^{\star})$ are not used by $\sigma$, so by definition of $\rho$ this is also the case for the leaves of $Z \cap V(T(v')$. Note that, by the choice of $u'$,  the whole subtree $T(u')$ remains intact throughout $\sigma'$. 
Since such a vertex $u'$ is $(T,T')$-good because no vertex in $T(v')$ is used by $\sigma$, we have that $\parent(T,u')= \parent(T',u')$. The fact that $\tau(T,Z,v^{\star})=\tau(T,Z,v')$ implies that the subtree $T(u')$ follows in $\sigma'$ the same moves within $T(v')$ as the moved followed by $T(u)$ in $\sigma$ within $T(v^{\star})$, where $u \in V(T^{\star})$ is a vertex such that $\tau(T,Z,u)=\tau(T,Z,u')$, which exists by the recursive definition of type and the hypothesis that $\tau(T,Z,v^{\star})=\tau(T,Z,v')$. Thus, $\parent(T,u')= \parent(T',u') = \parent(\hat{T},u')$.

Finally, we consider vertices $u'$ bottom-up in $V(T(v')) \cap Z$, assuming inductively that their strict descendants in $V(T(v') \cap Z$ already have their desired parent in $\hat{T}$ and exploiting the recursive definition of type. When encountering such a vertex $u'$, we further distinguish three cases.

\begin{itemize}
  \item If neither $u'$ nor $\parent(T,u')$ are used by $\sigma'$, since no vertex in $T(v')$ is used by $\sigma$, it follows that $\parent(T,u')= \parent(T',u') = \parent(\hat{T},u')$.
  \item If $u'$ is used by $\sigma'$, the analysis is similar to the case of $v,v'$ in the proof of \autoref{claim:case1-good-rotation-sequence} (cf.~\autoref{fig:mainlemma-case1}), by replacing the roles of $v$ and $v'$ in \autoref{claim:case1-good-rotation-sequence}, respectively, by $u$ and $u'$, where $u \in V(T^{\star})$ is a vertex such that $\tau(T,Z,u)=\tau(T,Z,u')$. Note that $u$ is used by $\sigma$ and that $\rho(u)=u'$. We can recursively assume that for all strict descendants $z$ of $u'$, $\parent(T',z) = \parent(\hat{T},z)$. The fact that $\tau(T,Z,u) = \tau(T,Z,u')$ implies that $u'$ follows in $\sigma'$ within $T(v')$ the same moves that $u$ follows in $\sigma$ within $T(v^{\star})$. And since $\wantparent(T,T',u) = \wantparent(T,T',u')$ (recall that the function $\wantparent(T,T',\cdot)$ is part of the definition of type), including the case where both sets are empty (meaning that they already have their respective desired parents), and since the whole subtree $T(u')$ remained intact throughout $\sigma$,
      it follows that $\parent(T',u') = \parent(\hat{T},u')$.

 \item Otherwise, if $u'$ is not used by $\sigma'$ but $\parent(T,u')$ is used by $\sigma'$, we apply the same arguments as in the case where $u'$ is a lowest such a vertex as discussed above, by replacing the property that ``the whole subtree $T(u')$ remains intact throughout $\sigma'$'' with ``for every strict descendant $z$ of $u'$, it holds that $\parent(T',z) = \parent(\hat{T},z)$'', which we can recursively assume. Thus, we conclude in the same way that $\parent(T,u') = \parent(T',u') = \parent(\hat{T},u')$.
\end{itemize}
\end{itemize}
The proof of the claim is now complete.
\end{cproof}

By \autoref{claim:case2-good-rotation-sequence},
$\sigma'$ is an $\ell$-rotation sequence from $T$ to $T'$, and it uses strictly more marked vertices than $\sigma$, because no vertex of $T(v^{\star})$ is marked (by the conditions in \autoref{claim:vertex-vstar-exists}), and within $T(v')$ there is at least one marked vertex used by $\sigma'$, namely $\rho(v)$ (cf. \autoref{fig:mainlemma-case2}). This concludes Case~2.
O

\medskip

In both cases, we have defined from $\sigma$ another $\ell$-rotation sequence $\sigma'$ from $T$ to $T'$ using strictly less vertices in $V(T) \setminus M$ than $\sigma$, contradicting the choice of $\sigma$ and concluding the proof of the lemma.
\end{proof}

\subsection{Wrapping up the algorithm}
\label{sec:wrapping-up}

We finally have all the ingredients to prove our main result, which we restate for convenience.

\maintheorem*

\begin{proof}
Given a connected graph $G$, two elimination trees $T$ and $T'$ of $G$, and a positive integer $k$ as input of the \kelimination problem, we proceed as follows.

By \autoref{prop:+3kno}, we can assume that our instance contains at most $3k$ $(T,T')$-children-bad vertices, which can be clearly identified in time linear in $|V(G)|$. Recall from \autoref{def:ball-around-children-bad} that, if we
let $C \subseteq V(T)$ be the set of $(T,T')$-children-bad vertices,
$\Bcb = N_T^{2k}[C \cup \root(T)]$. By \autoref{lemma:restriction-to-balls},
If $\dist(T,T')\leq k$, then there exists an $\ell$-rotation sequence from $T$ to $T'$, for some $\ell \leq k$, using only vertices in $\Bcb$.

We now apply our marking algorithm to find the set $M \subseteq \Bcb$ of marked vertices. By \autoref{lem:M-bounded}, the set $M$ has size at most $h(k)$ and can be computed in time $h(k) \cdot |V(G)|$,
where  $h(k)$ has the same asymptotic growth as the function $g(k)$ given by \autoref{lem:number-types}. By \autoref{lem:main}, if $\dist(T,T')\leq k$, then there exists an $\ell$-rotation sequence from $T$ to $T'$, for some $\ell \leq k$, using only vertices in $M$.

Thus, we can solve the problem by applying the following naive brute force algorithm: for every $\ell \in [k]$ (we may assume that $T$ and $T'$ are distinct), we guess all possible sets of $\ell$ ordered pairs of vertices in $M$ (so, $2 \ell$ vertices overall, allowing repetitions), and for each such an ordered set of $\ell$ pairs, we apply the corresponding rotations to $T$ and check whether the resulting elimination tree is equal to $T'$ or not, which can be done in time linear in $k \cdot |V(G)|$. Naturally, if for some of the guessed pairs to be rotated, that edge does not exist in the corresponding intermediate elimination tree of $G$, we discard that guess.

The running time of the resulting algorithm is upper-bounded by $\Ocal(k \cdot |M|^{2k} \cdot |V(G)|)$, and the theorem follows.
\end{proof}

\section{Hardness results on hypergraphic polytopes}
\label{sec:hardness}

In this section we provide hardness results for computing distances on hypergraphic polytopes. We first introduce hypergraphic polytopes in \autoref{sec:hypergraphic}, and restate the problems in terms of orientations of hypergraphs. In \autoref{sec:W2} we prove that computing distances on hypergraphic polytopes  is ${\sf W}[2]$-hard parameterized by the distance. In
\autoref{sec:inapprox} we prove inapproximability results, and in \autoref{sec:no-kernels} we rule out the existence of polynomial kernels parameterized by the number of vertices of the hypergraph.

 \subsection{Distances on hypergraphic polytopes}
 \label{sec:hypergraphic}

We first need to introduce some definitions, mostly taken from the recent paper of Cardinal and Steiner~\cite{CardinalS25}. Every hypergraph $H = (V, \mathcal{E})$, where $\mathcal{E} \subseteq 2^V \setminus \{\emptyset\}$, induces a  submodular function $f_H$ defined such that, for every $U \subseteq V$,
$$f_H(U) := |\{e \in \mathcal{E} : U \cap e\neq \emptyset\}|.$$

The base polytope induced by the polymatroid associated with $f_H$ is called the \emph{hypergraphic polytope} of the hypergraph $H$, and denoted by $P_H$.
When $H$ is a graph, the polymatroids are the graphical zonotopes, whose vertices are in bijection with acyclic orientations of the graph. Graph associahedra are also a particular type of hypergraphic polytopes, as discussed for instance in~\cite{CardinalS25}. We refer the reader to~\cite{CardinalS25} and the references therein for more context and details about hypergraphic polytopes; see also~\cite{BBM19,Hel74,VW93,AA23,PRW08,Post09,Reh22}.

It is known that the vertices of the hypergraphic polytope $P_H$ are in one-to-one
correspondence with so-called acyclic orientations of the hypergraph $H$, defined as follows. An \emph{orientation} of a hypergraph $H = (V, \mathcal{E})$ is a mapping $h : \mathcal{E} \rightarrow V$ such that $h(e) \in e$ for every $e \in \mathcal{E}$.
The element $h(e)$ is called the \emph{head} of an edge $e$. With every orientation $h$ of $H$ we can associate a directed graph $D_h$ that has vertex set $V$ 
and an arc from a vertex $u$ to a vertex $v$ if and only if there exists some $e \in \mathcal{E}$ such that $h(e) = v$ and $u \in e \setminus {v}$.
Intuitively, $D_h$ is the union, over all $e \in \mathcal{E}$, of the inwards-oriented stars centered at the head of $e$ and whose leaves are all other elements of $e$. With this notion at hand, we call an orientation $h$ of $H$ \emph{acyclic} if $D_h$ is an acyclic digraph.

Given two distinct acyclic orientations $h_1$, $h_2$ of a hypergraph $H = (V, \mathcal{E})$, we say that $h_2$ is obtained from $h_1$ by a \emph{flip} at an ordered pair $(u, v) \in V^2$ of distinct vertices, if $h_2$ is obtained from $h_1$ by redefining the head of every hyperedge $e \in \mathcal{E}$ with $e \supseteq \{u, v\}$ that satisfies $h_1(e) = u$ as $h_2(e) = v$, while
keeping all other head assignments the same as in $h_1$. We note that the above definition of a flip is symmetric, in the following sense: If $h_2$ can be obtained from $h_1$ by flipping at $(u, v)$, then $h_1$ is obtained from $h_2$ by flipping at $(v, u)$.

Also, note that when $h_2$ is obtained from $h_1$ by flipping at $(u, v)$, then a flip at
$(v, u)$ for $h_1$ is impossible (due to acyclicity). Thus, in the following we can say without ambiguity
that two acyclic orientations $h_1$ and $h_2$ are obtained from each other by flipping
the unordered pair $\{u, v\}$ of distinct vertices if $h_2$ can be obtained from $h_1$ by
flipping $(u, v)$ or $(v, u)$. A nice illustration of a flip can be found in~\cite[Figure 4.1]{CardinalS25}.

The following statement can be deduced as a special case of~\cite[Theorem 2.18]{BBM19}.

\begin{proposition}\label{prop:equivalence-hypergraphic-polytopes}
Let $H = (V, \mathcal{E})$ be a hypergraph. Two vertices $x, y$ of the polytope
$P_H$ corresponding to acyclic orientations $h_1, h_2$ of $H$ are adjacent on the skeleton
of $P_H$ if and only if there exists a pair $\{u, v\} \subseteq V$ of distinct vertices such that $h_1$ and $h_2$ can be obtained from each other by flipping $\{u, v\}$.
\end{proposition}

\autoref{prop:equivalence-hypergraphic-polytopes} implies that the computation of shortest paths on the skeleton of a hypergraphic
polytope $P_H$ is equivalent to a flip distance problem between acyclic orientations of the corresponding hypergraph $H$, formally defined as follows.

\medskip
\defproblema{Flip Distance Between Acyclic Orientations}
{A hypergraph $H = (V, \mathcal{E})$, two acyclic orientations $h_1$, $h_2$ of $H$, and a positive integer $\ell$.}
{Is there a sequence of vertex pairs $\{u_i, v_i\}$, $i  \in [\ell]$ such that $h_2$ can be obtained from $h_1$ by flipping, in order, the pairs $\{u_1, v_1\},\ldots,\{u_{\ell},v_{\ell}\}$?}

In the next subsections we provide hardness results for the above problem.

\subsection{W[2]-hardness}
\label{sec:W2}

In our next result we prove that a generalization of \autoref{thm:main} to hypergraphic polytopes is unlikely.

\begin{theorem}\label{thm:w2hard}
The {\sc Flip Distance Between Acyclic Orientations} problem is ${\sf W}[2]$-hard when parameterized by the distance $\ell$, even restricted to hypergraphs with polynomially many hyperedges in terms of the number of vertices.
\end{theorem}
\begin{proof}
The proof is a parameterized reduction from {\sc Dominating Set} parameterized by the size of the solution, which is well known to be ${\sf W}[2]$-hard~\cite{DoFe13}.

Let $G$ be a simple graph and $k$ be an integer as input of {\sc Dominating Set}, that is, the question is whether $G$ admits a dominating set of size at most $k$. We construct from $G$ a hypergraph $H$ as follows:
\begin{itemize}
    \item $S:=\{s_i \mid 1\leq i\leq k^2\}$.
    \item $T:=\{t_j \mid 1\leq j\leq k^2\}$.
    \item $V(H):=V(G)\cup S\cup T$.
    \item $\mathcal{E}(H):=\{N[v]\cup \{s_i,t_j\} \mid v\in V(G), i,j\in[1,k^2]\}$.
\end{itemize}

For the sake of the construction of the above hypergraph, we can assume that the input graph $G$ contains no two distinct vertices $u,v$ with $N[u] = N[v]$ (as if it does, one of them can be safely deleted without changing the type of the instance). Under this assumption, note that all the hyperedges in $\mathcal{E}(H)$ are distinct. Also, note that $|\mathcal{E}(H)|$ is indeed polynomially bounded as a function of $|V(H)|$.

Now, consider the following acyclic orientations of $H$:
\begin{itemize}
    \item $h_1$: for each $e\in \mathcal{E}(H)$, $h(e):=e\cap S$.
    \item $h_2$: for each $e\in \mathcal{E}(H)$, $h(e):=e\cap T$.
\end{itemize}

For $i \in \{1,2\}$, let $D_{h_i}$ be the directed graph associated with $H$ with respect to $h_i$. Note that, by construction, $D_{h_i}$ is indeed acyclic for $i \in \{1,2\}$. We define an instance of {\sc Flip Distance Between Acyclic Orientations} as $H, h_1, h_2, \ell := 2k^3$.

We shall now prove that $G$ has a dominating set of size at most $k$ if and only if one can obtain $h_2$ from $h_1$ by flipping at most $2k^3$ pairs of vertices of $H$. The main idea of the reduction is that, to go from $h_1$ to $h_2$, instead of changing the head of the hyperedges by flipping all pairs $\{s_i,t_j\}$ for $i,j \in [k^2]$, it is more efficient to ``pivot'' through the vertices in a dominating set $D$ of $G$, namely by first setting the vertices in $D$ as heads of all hyperedges (which originally had heads in $S$), and then flipping the heads from $D$ to the target heads in $T$; see \autoref{fig:hardness} for an illustration. Along this sequence of flips, we need to be careful with the order in which they are performed, so that all intermediate orientations of $H$ are also acyclic, as required in the definition of the problem. The details follow.

\begin{figure}[h!tb]
    \centering
    \includegraphics[width=1.05\textwidth]{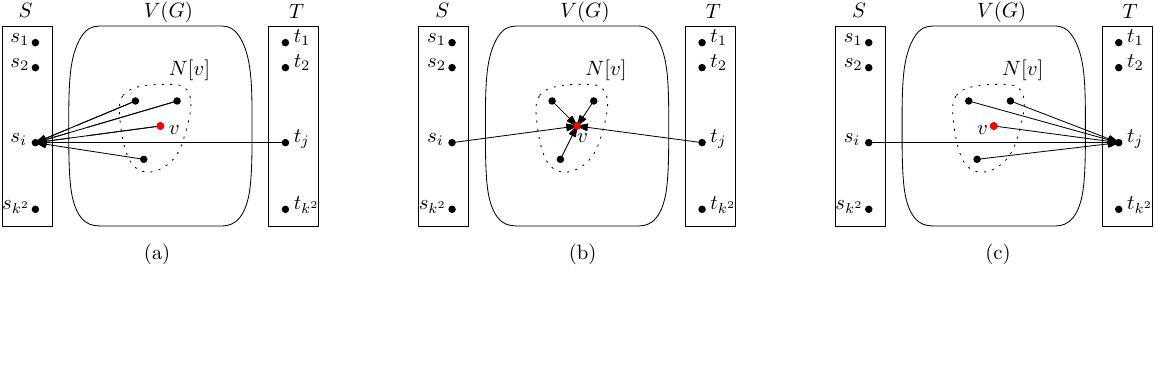}
    \vspace{-1.7cm}
    \caption{Illustration of the reduction of \autoref{thm:w2hard}. A generic hyperedge $e$ of $H$ with vertices $N[v] \cup \{s_i,t_j\}$ is shown, together with the following orientations: (a) initial orientation $h_1$. (b) intermediate orientation pointing towards the vertices in the dominating set $D$, assuming that $v \in D$. (c) target orientation $h_2$.\label{fig:hardness}}
\end{figure}

\medskip

($\Rightarrow$) Let $D=\{b_1,b_2,\ldots,b_k\}$ be a dominating set of $G$ of size $k$ (for simplicity, we can assume that the dominating set has size {\sl exactly} $k$). We define the following sequence of $\ell = 2k^3$ flips from $h_1$ to $h_2$, while guaranteeing that  each of the intermediate orientations of $H$ is also acyclic:
\begin{itemize}
    \item Flip $\{s_i,b_r\}$ for each $s_i\in S$ and $b_r\in D$ in lexicographic order, that is, we first do the flips $\{s_1,b_r\}$ for $r = 1, \ldots k$, then $\{s_2,b_r\}$ for $r = 1, \ldots, k$, and so on ($k^3$ flips in total).
    \item Flip $\{b_r,t_j\}$ for each $b_r\in D$ and $t_j\in T$ in reverse lexicographic order), that is, we first do the flips $\{b_k,t_j\}$ for $j = k^2, \ldots, 1$, then $\{b_{k-1},t_j\}$ for $j = k^2, \ldots, 1$, and so on ($k^3$ flips in total).
\end{itemize}

Clearly, the above sequence performs $2k^3$ flips. In addition, since $D$ is a dominating set of $G$, it holds that $e\cap D\neq \emptyset$ for each hyperedge of $H$. Thus, the first $k^3$ flips change $h(e)$ for each $e$ setting it to some vertex of $D$. After that, the last $k^3$ flips change $h(e)$ for each $e$ setting it to some vertex of $T$, as desired. To conclude this direction of the equivalence, it just remains to verify that, indeed, all intermediate orientations are acyclic.

\begin{claim}\label{claim:all-acyclic}
Along the above sequence of flips, all intermediate orientations of $H$ remain acyclic.
\end{claim}
\begin{cproof} We verify the claim by providing topological orderings of $V(H)$ such that all arcs in the corresponding orientations go from right to left. It is well known that the existence of such a topological ordering of a digraph is equivalent to it being acyclic.

Let us start with the first sequence of $k^3$ flips. Initially, in $D_{h_1}$, where the vertices in $S$ are sinks and all other vertices are sources, we consider the following topological ordering:
$$
s_{k^2}, \ldots, s_2, s_1, b_1, b_2, \ldots,  b_k, V(G) \setminus D, T,
$$

where we can order the vertices in  $V(G) \setminus D$ and $T$ arbitrarily. Now, each of the flips swaps the corresponding pair of vertices in the above ordering. Namely, after the first flip  $\{s_1,b_1\}$, the resulting topological ordering is (the pair of swapped vertices is shown in bold)
$$
s_{k^2}, \ldots, s_2, \boldsymbol{b_1},\boldsymbol{s_1}, b_2, \ldots,  b_k, V(G) \setminus D, T,
$$
then $s_1$ is swapped with $b_2$ and so on, and after the flip $\{s_1,b_k\}$ the corresponding ordering is
$$
s_{k^2}, \ldots, s_2, b_1, b_2, \ldots,  \boldsymbol{b_k}, \boldsymbol{s_1}, V(G) \setminus D, T.
$$
Intuitively, we have shifted $s_1$ from the left to the right of the set $D = \{b_1, \ldots, b_k\}$. Then the same sequence of swaps is triggered for $s_2$, and so on, until all vertices of $S$ are swapped, resulting in the ordering
$$
b_1, b_2, \ldots,  \boldsymbol{b_k}, \boldsymbol{s_{k^2}}, \ldots, s_2, s_1, V(G) \setminus D, T.
$$
Let us verify that all the above orderings are indeed topological. Consider an intermediate ordering resulting after a generic flip $\{s_i,b_r\}$, namely,
$$
s_{k^2}, \ldots ,s_{i+1}, b_1,b_2, \ldots, \boldsymbol{b_{r}},\boldsymbol{s_i}, b_{r+1},\ldots, b_k, s_{i-1},\ldots,s_1, V(G) \setminus D, T,
$$
and let us verify that, in the current orientation, all arcs go from right to left. Note that all vertices before $b_1$ are sinks, and all vertices after $b_k$ are sources, because $D$ is a dominating set of $G$. Moreover, vertex $s_i$ is a sink for all arcs from $b_{r+1}, \ldots, b_k$, and a source for all arcs to $b_1, \ldots, b_r$. Thus, it just remains to verify that there cannot be an arc from a vertex $b_p$ to a vertex $b_q$ with $p<q$. Suppose for contradiction that it is the case. Then this arc $(b_p,b_q)$ has been created by a flip $\{s_x,b_q\}$ for some $x \in [k^2]$, and necessarily $b_p$ and $b_q$ are both in the closed neighborhood of a vertex of $G$, so they both appear together in hyperedges of $H$. Let $e$ be one of such hyperedges. By the lexicographical order in which the flips are done, the flip $\{s_x,b_p\}$ was done before the flip $\{s_x,b_q\}$, and therefore the head of $e$ right before the flip $\{s_x,b_q\}$ was $b_p$, so the head of $e$ is not affected by the flip $\{s_x,b_q\}$, and we conclude that the arc $(b_p,b_q)$ cannot exist, a contradiction.

\medskip

Let us now focus on the second sequence of $k^3$ flips, the analysis being completely symmetric to the above one. Namely, we start with the following topological ordering of $V(H)$, observing that after the first sequence of flips, the vertices in $D$ are sinks, and all other vertices are sources (note that this ordering is different from the one considered at the end of the first sequence of flips):
$$
b_1, \ldots ,b_k, t_{k^2},\ldots,t_1, V(G) \setminus D, S,
$$
where we can order the vertices in  $V(G) \setminus D$ and $S$ arbitrarily. After the first flip  $\{b_k,t_{k^2}\}$, the resulting topological ordering is
$$
b_1, \ldots ,b_{k-1}, \boldsymbol{t_{k^2}},\boldsymbol{b_k}, t_{k^2 -1},\ldots, t_1, V(G) \setminus D, S.
$$
 Then, similarly to the first sequence of flips, $b_k$ is shifted until the right of $t_1$, then the same applies to $b_{k-1}$, and so on, until we obtain the final ordering
$$
t_{k^2},\ldots,\boldsymbol{t_1},\boldsymbol{b_1}, \ldots ,b_k, V(G) \setminus D, S,
$$
where all vertices in $T$ are sinks and all other vertices are sources. Similarly to the analysis for the first sequence of flips, one can verify that  all intermediate orientations of $H$ remain acyclic.
\end{cproof}

\medskip

($\Leftarrow$) Conversely, suppose now  that there exists a sequence $F$ of flips of pairs of vertices of $H$ going from $h_1$ to ${h_2}$, with $|F| \leq 2k^3$.

Note that there exists some $i \in [k^2]$ such that the pair $\{s_i,t_j\}$ is not in $F$, as otherwise the number of flips would be at least $k^4$, and assuming that $k\geq 3$ (otherwise \textsc{Dominating Set} can be solved in polynomial time), we have that $k^4 > 2k^3$. Therefore, for each $e=N[v]\cup \{s_i,t_j\}$ where $v\in V(G)$, since all orientations of  arcs incident with $s_i$ and $t_j$ need to change, there is a pair $\{s_i,b\}$ in $F$ where $b\in N[v]$, and there is as well a pair $\{b,t_j\}$ in $F$.

Let $D=\{b \in V(G)\mid \{s_i,b\}\in F\}$. It follows that $D$ is a dominating set of $G$, because in $h_1$, vertex $s_i$ has all arcs incoming from the whole $V(G)$, and in $h_2$, vertex $s_i$ has no arcs at all from/to $V(G)$. To conclude the proof, it remains to show that $|D|\leq k$.

To this end, let $p_1$ be the number of pairs $\{s_i,t_j\}\in F$, and let $p_2=k^4-p_1$. Let also $\alpha$ (resp. $\beta$) be the number of vertices $s_i \in S$ (resp. $t_j \in T$) with some pair $\{s_i,t_j\}\notin F$, and note that, by definition, $\alpha\cdot \beta \geq p_2$. Then it holds that
\begin{equation}\label{eq:F}
|F|\geq p_1+\alpha\cdot |D| + \beta\cdot |D| = p_1 + (\alpha + \beta) \cdot |D|.
\end{equation}


\begin{claim}\label{claim:minimized}
$\alpha + \beta \geq 2 \sqrt{p_2}$.
\end{claim}
\begin{cproof}
Write $\alpha = \sqrt{p_2}+c$ for some real number $c \geq 0$. Since $\alpha\cdot \beta \geq p_2$, it follows that
$$
\alpha + \beta = \sqrt{p_2}+c + \beta \geq \sqrt{p_2}+c + \frac{p_2}{\sqrt{p_2}+c} \geq 2 \sqrt{p_2}.
$$
\end{cproof}
%
%
%
%
%
%
Using \autoref{claim:minimized} in \autoref{eq:F} we get that
\begin{equation}\label{eq:F2}
|F|\geq  p_1 + 2 \sqrt{p_2} \cdot |D|.
\end{equation}
Using that $|F|\leq 2k^3$ and that  $p_1 + p_2 = k^4$, from \autoref{eq:F2} we get that
\begin{equation}\label{eq:Bbound}
|D|\leq  \frac{|F| - p_1}{2 \sqrt{p_2}} \leq \frac{2k^3 - p_1}{2 \sqrt{p_2}} = \frac{2k^3 - p_1}{2 \sqrt{k^4 - p_1}}.
\end{equation}
We will conclude the prove by using the following simple monotonicity argument.
\begin{claim}\label{claim:non-increasing}
For $k\geq 3$, the function $\frac{2k^3 - p_1}{2 \sqrt{k^4 - p_1}}$ of \autoref{eq:Bbound} is decreasing, with $p_1$ being the variable, for the range of allowed values $0 \leq p_1 \leq 2k^3$.
\end{claim}
\begin{cproof}
The derivative of $\frac{2k^3 - p_1}{2 \sqrt{k^4 - p_1}}$ with respect to $p_1$ equals
\begin{equation}\label{eq:derivative1}
 \frac{p_1 + 2k^3 -2k^4}{4(k^4-p_1)^{3/2}},
\end{equation}
whose denominator is strictly positive because, for $k \geq 3$, $p_1 \leq 2k^3 < k^4$. On the other hand, the numerator satisfies
\begin{equation*}
p_1 + 2k^3 -2k^4 \leq 4k^3 -2k^4 = 2(2k^3-k^4)<0.
\end{equation*}
Thus, the function of \autoref{eq:derivative1} is strictly negative for $0 \leq p_1 \leq 2k^3$, implying that the function considered in the statement of the claim is decreasing.
\end{cproof}
By \autoref{claim:non-increasing}, the function on the right-hand size  of \autoref{eq:Bbound} can be upper-bounded by its value for $p_1=0$, which gives
\begin{equation*}\label{eq:B2}
|D|\leq  \frac{2k^3 - p_1}{2 \sqrt{k^4 - p_1}} \leq \frac{2k^3}{2 \sqrt{k^4}} = k,
\end{equation*}
concluding the proof of the theorem.
\end{proof}

\autoref{prop:equivalence-hypergraphic-polytopes} yields the following direct corollary of \autoref{thm:w2hard}.

\begin{corollary}
 Let $H = (V, \mathcal{E})$ be a hypergraph, let $v_{h_1}, v_{h_2}$ be two vertices corresponding to the acyclic orientations $h_1, h_2$ of $H$ in the polytope $P_H$, respectively, and let  $\ell$ be a positive integer. The problem of determining whether $v_{h_1}$ is at distance at most $\ell$ from $v_{h_2}$ in the polytope $P_H$ is ${\sf W}[2]$-hard when parameterized by $\ell$.
 \end{corollary}



\subsection{Inapproximability}
\label{sec:inapprox}

In this section we slightly modify the construction of the previous section to obtain inapproximability results for  {\sc Flip Distance Between Acyclic Orientations}. We rely on the result of Raz and Safra~\cite{RazS97} that \textsc{Dominating Set} cannot be approximated in polynomial time within a factor $d \cdot \log n$ for some constant $d > 0$ unless $\P = \NP$, where $n$ is the number of vertices of the input graph.

It is worth mentioning that we could also use the stronger lower bound for \textsc{Dominating Set} of $(1-\varepsilon)\cdot \ln n$ for any $\varepsilon > 0$ by Feige~\cite{Feige98} that relies on the stronger hypothesis $\NP \not\subseteq {\sf DTIME}(n^{\Ocal(\log\log n})$, and obtain a similar lower bound for the {\sc Flip Distance Between Acyclic Orientations} problem.

\begin{theorem}\label{thm:inapproximability}
  There exists some constant $c > 0$ such that the {\sc Flip Distance Between Acyclic Orientations} problem does not admit a polynomial-time $c \cdot \log (|V|+|{\mathcal E}|)$-approximation algorithm unless $\P = \NP$, where $H=(V,{\mathcal E})$ is the input hypergraph.
\end{theorem}
\begin{proof}
Let $G=(V,E)$ be the input of the \textsc{Dominating Set} problem, and let $n := |V(G)|$. By the result of Raz and Safra~\cite{RazS97}, there exists some constant $d > 0$ such that \textsc{Dominating Set} cannot be approximated in polynomial time within a factor $d \cdot \log n$, unless $\P = \NP$. We proceed to construct an instance $H,h_1,h_2$ of the optimization version of {\sc Flip Distance Between Acyclic Orientations}, that is, without the integer $\ell$ in the input. The reduction is almost the same as in the proof of \autoref{thm:w2hard}, except that the sets $S$ and $T$ now have size $n^2$ instead of $k^2$. That is, we let
\begin{itemize}
    \item $S:=\{s_i \mid 1\leq i\leq n^2\}$.
    \item $T:=\{t_j \mid 1\leq j\leq n^2\}$.
    \item $V(H):=V(G)\cup S\cup T$.
    \item $\mathcal{E}(H):=\{N[v]\cup \{s_i,t_j\} \mid v\in V(G), i,j\in[1,n^2]\}$.
\end{itemize}

We consider the following two acyclic orientations of $H$:
\begin{itemize}
    \item $h_1$: for each $e\in \mathcal{E}(H)$, $h(e):=e\cap S$.
    \item $h_2$: for each $e\in \mathcal{E}(H)$, $h(e):=e\cap T$.
\end{itemize}

Note that and that $|V(H)| = |V(G)| + |S| + |T| = 2n^2 + n$,  and that $|\mathcal{E}(H)| = |V(G)| \cdot |S| \cdot |T| = n^5$.   Let $\gamma$ be the domination number of $G$, that is, the smallest size of a dominating set in $G$. By the same analysis as in the proof of \autoref{thm:w2hard}, it is easy the verify that the optimal solution (that is, a flip sequence) of the {\sc Flip Distance Between Acyclic Orientations} with input $H,h_1,h_2$ has size  ${\sf opt}(H,h_1,h_2)= 2\gamma \cdot n^2$.

 Let $c := d / 6$, where $d$ is the inapproximability constant for \textsc{Dominating Set} defined above. To conclude the proof, we shall prove that if there exists a polynomial-time approximation algorithm with ratio $c \cdot \log (|V(H)|+|{\mathcal E}(H)|)$ for the {\sc Flip Distance Between Acyclic Orientations} problem, then there exists a polynomial-time approximation algorithm with ratio $d \cdot \log n$ for \textsc{Dominating Set}.

 To this end, suppose that we can find in polynomial time (in the size of $H$) a flipping sequence $F$ from $h_1$ to $h_2$ such that
 \begin{align}
   |F|  \ \leq \  {\sf opt}(H,h_1,h_2) \cdot c \cdot \log (|V(H)|+|{\mathcal E}(H)|)  \  = \  2\gamma \cdot n^2 \cdot c \cdot \log ( 2n^2 + n + n^5) \\
 \label{eq:F-upper-bound} \leq \  2\gamma \cdot n^2 \cdot c \cdot \log ( 4n^5)\   \leq\   2\gamma \cdot n^2 \cdot 6c \cdot \log n\  = \ 2\gamma \cdot n^2 \cdot d \cdot \log n,\hspace{1.63cm}
 \end{align}

where in the last inequality we have used that $n \geq 4$, which we can clearly assume.

  Similarly to the proof of \autoref{thm:w2hard}, that there exists some $i \in [n^2]$ such that the pair $\{s_i,t_j\}$ is not in $F$, as otherwise the number of flips would be at least $n^4$, which is larger than the upper bound given in \autoref{eq:F-upper-bound}, for $n$ large enough, because $\gamma \leq n$.
Therefore, for each $e=N[v]\cup \{s_i,t_j\}$ where $v\in V(G)$, since all orientations of  arcs incident with $s_i$ and $t_j$ need to change, there is a pair $\{s_i,b\}$ in $F$ where $b\in N[v]$, and there is as well a pair $\{b,t_j\}$ in $F$.
Let $D=\{b \in V(G)\mid \{s_i,b\}\in F\}$, which can be clearly computed in polynomial time. It follows that $D$ is a dominating set of $G$. To conclude the proof, it remains to show that $|D| \leq \gamma \cdot d \cdot \log n$.

To this end, let again $p_1$ be the number of pairs $\{s_i,t_j\}\in F$, and let $p_2=n^4-p_1$. With the same argument as in the proof of \autoref{thm:w2hard} (cf. \autoref{claim:minimized} and \autoref{eq:F2}) we can prove that
 \begin{equation}\label{eq:F2bis}
|F|\geq  p_1 + 2 \sqrt{p_2} \cdot |D|.
\end{equation}
 Using the fact that $p_1 + p_2 = n^4$, from \autoref{eq:F2bis} we get that
\begin{equation}\label{eq:Bbound2}
|D|\leq  \frac{|F| - p_1}{2 \sqrt{p_2}}  = \frac{|F| - p_1}{2 \sqrt{n^4 - p_1}}.
\end{equation}
The following claim is very similar to \autoref{claim:non-increasing}.
\begin{claim}\label{claim:non-increasing2}
For $n$ large enough, the function $\frac{|F| - p_1}{2 \sqrt{n^4 - p_1}}$ of \autoref{eq:Bbound2} is decreasing, with $p_1$ being the variable, for the range of allowed values $0 \leq p_1 \leq |F|$.
\end{claim}
\begin{cproof}
The derivative of $\frac{|F| - p_1}{2 \sqrt{n^4 - p_1}}$ with respect to $p_1$ equals
\begin{equation}\label{eq:derivative}
\frac{p_1 + |F| -2n^4}{4(n^4-p_1)^{3/2}},
\end{equation}
whose denominator is strictly positive because, using the upper bound on $|F|$ given by \autoref{eq:F-upper-bound}, we get that  $p_1 \leq |F| \leq 2\gamma \cdot n^2 \cdot d \cdot \log n \leq 2d \cdot \log n \cdot n^3 < n^4$, where the last inequality holds provided that $n$ is large enough. On the other hand, the numerator satisfies
\begin{equation*}
p_1 + |F| -2n^4 \leq 2|F| -2n^4 = 2(|F|-n^4)<0,
\end{equation*}
where the last inequality follows by the same calculation used to show that the denominator is positive. Thus, the function of \autoref{eq:derivative} is strictly negative for $0 \leq p_1 \leq |F|$, implying that the function considered in the statement of the claim is decreasing.
\end{cproof}
By \autoref{claim:non-increasing2}, the function on the right-hand size of  \autoref{eq:Bbound2} can be upper-bounded by its value for $p_1=0$, which gives, using the upper bound on $|F|$ given by \autoref{eq:F-upper-bound}, that
\begin{equation*}\label{eq:B2bis}
|D|\leq  \frac{|F| - p_1}{2 \sqrt{n^4 - p_1}} \leq \frac{|F|}{2 \sqrt{n^4}} = \frac{|F|}{2 n^2} \leq \frac{2\gamma \cdot n^2 \cdot d \cdot \log n}{2 n^2} = \gamma \cdot d \cdot \log n,
\end{equation*}
concluding the proof of the theorem.
\end{proof}

\subsection{Ruling out polynomial kernels}
\label{sec:no-kernels}

Given the ${\sf W}[2]$-hardness result of \autoref{thm:w2hard}, it is natural to identify parameterizations of the {\sc Flip Distance Between Acyclic Orientations} problem that allow for \FPT algorithms. On such a parameterization is the number of vertices of the input hypergraph. Indeed, if the input hypergraph is $H = (V, \mathcal{E})$, it holds that $|\mathcal{E}| \leq 2^{|V|}$. On the other hand,  Cardinal and Steiner~\cite[Lemma 6.1]{CardinalS25}) proved that if $h_1$, $h_2$ are distinct acyclic orientations of a hypergraph $H = (V, \mathcal{E})$, then there exists $e \in \mathcal{E}$ such that $h_1(e) \neq h_2(e)$ and $(h_1(e), h_2(e))$ is flippable in $h_1$. This result implies that it is always possible to reach the target orientation $h_2$ from the initial orientation $h_1$ by doing at most $|V|^2$ flips, namely by fixing all ``wrong'' pairs of heads one by one, and observing that once a pair is fixed, then it remains fixed until the end of the algorithm. The above discussion implies that the problem is \FPT parameterized by $|V|$. Thus, it makes sense to ask whether it admits a polynomial kernel parameterized by $|V|$. In our next result we prove that it is probably not the case. The reduction is very similar to the one given in \autoref{thm:w2hard}, but instead of reducing from \textsc{Dominating Set} parameterized by the size of the solution, we reduce from \textsc{Red-Blue Dominating Set} parameterized by the size of the ``red'' set in the bipartition, which is known not to admit polynomial kernels unless ${\sf NP} \subseteq {\sf coNP}/{\sf poly}$~\cite{DomLS14,FominLSZ19}.

\begin{theorem}\label{thm:no-kernels}
The {\sc Flip Distance Between Acyclic Orientations} problem does not admit polynomial kernels parameterized by the number of vertices of the input hypergraph, unless ${\sf NP} \subseteq {\sf coNP}/{\sf poly}$.
\end{theorem}
\begin{proof}
We present a polynomial parameter transformation (PPT) from the \textsc{Red-Blue Dominating Set} problem parameterized by the size of the ``red'' set, which is know not to admit polynomial kernels unless ${\sf NP} \subseteq {\sf coNP}/{\sf poly}$~\cite{DomLS14,FominLSZ19}. The input of \textsc{Red-Blue Dominating Set} is a bipartite graph $G = (R \cup  B, E)$, where $R$ is the set of \emph{red} vertices and $B$ is the set of \emph{blue} vertices, and an integer
$k$, and the goal is a to determine whether there exists a set $D\subseteq R$ of at most $k$ red vertices that dominates all the
blue vertices (i.e., each blue vertex is adjacent to some vertex in $D$). The parameter is $|R|$. We can assume that $|B| \geq |R|$, as otherwise the problem trivially admits a polynomial kernel parameterized by $|R|$. We can also assume that $k \leq |R|$, as otherwise the answer is clearly ``{\sf yes}''. We construct from $G = (R \cup  B, E),k$ an instance of {\sc Flip Distance Between Acyclic Orientations} as follows. The reduction is very similar to the one of \autoref{thm:w2hard}.

Namely, we construct from $G$ a hypergraph $H$ as follows:
\begin{itemize}
    \item $S:=\{s_i: 1\leq i\leq k^2\}$.
    \item $T:=\{t_j: 1\leq j\leq k^2\}$.
    \item $V(H):=R\cup S\cup T$.
    \item $\mathcal{E}(H):=\{N(v)\cup \{s_i,t_j\}: v\in B, i,j\in[1,k^2]\}$.
\end{itemize}

That is, the vertex set of the constructed hypergraph $H$ corresponds to the {\sl red} vertices together with the sets $S$ and $T$, and the hyperedges correspond to the (open) neighborhoods of the {\sl blue} vertices together with one vertex of $S$ and one vertex of $T$.

We define the initial and final acyclic orientations of $H$ in the same way as in the proof of \autoref{thm:w2hard}:
\begin{description}
    \item[$h_1$:] for each $e\in \mathcal{E}(H)$, $h(e):=e\cap S$;
    \item[$h_2$:] for each $e\in \mathcal{E}(H)$, $h(e):=e\cap T$;
\end{description}

The same argument used in the proof \autoref{thm:w2hard} shows that $G$ contains
a set $D\subseteq R$ of at most $k$ red vertices that dominates all the
blue vertices if and only if one can obtain $h_2$ from $h_1$ by flipping at most $2k^3$ pairs of vertices of $H$. Since $|V(H)| = 2k^2 + |R| \leq 2|R|^2+|R|$, this reduction is indeed a PPT from \textsc{Red-Blue Dominating Set} parameterized by $|R|$ to {\sc Flip Distance Between Acyclic Orientations} parameterized by $|V(H)|$, and the theorem follows.
\end{proof}


\section{Further research}
\label{sec:discussion}


We start with further research about computing distances on graph associahedra, and then on hypergraphic polytopes.

\medskip

We proved that the  \kelimination problem, for a general graph $G$, can be solved in time $f(k) \cdot |V(G)|$, where $f(k)$ is the function given by \autoref{thm:main}. This function is quite large, and it is worth trying to improve it. The growth of $f(k)$ is mainly driven by the number of different types of vertices (cf. \autoref{def:type}) that we consider in our marking algorithm. We need this recursive definition of type to guarantee that, when two vertices $v,v'$ have the same type, then for each possible type $\tau$ and every integer $d$ at most the bound given in \autoref{eq:bounded-diameter}, vertices $v$ and $v'$ have the same number (up to $k+1$) of descendants of  type $\tau$ within distance $d$. This is exploited, for instance, in Case 2 of the proof of \autoref{lem:main} to apply a recursive argument. It may possible to find a simpler argument in the replacement operation performed in the proof of \autoref{lem:main} (using the representative function $\rho)$, and in that case, one may allow for a less refined notion of type, leading to a better bound.

Another natural direction is to investigate whether \kelimination admits a polynomial kernel parameterized by $k$. So far, this is only known when the considered graph $G$ is a path, where even linear kernels are known\cite{cleary2009rotation,Lucas10}. As an intermediate step, one may consider graphs of bounded degree, for which it seems plausible that \autoref{lemma:restriction-to-balls} (restriction to few balls of bounded diameter) provides a helpful opening step.

\medskip

There are also several interesting questions about computing distances on hypergraphic polytopes. We proved in \autoref{thm:w2hard} that
the problem is ${\sf W}[2]$-hard when parameterized by the target distance $\ell$. It makes sense to identify additional parameters that allow for \FPT algorithms. A natural one is the \emph{rank} of the hypergraph, that is, the maximum size of a hyperedge. Note that in the reduction of \autoref{thm:w2hard}, the rank of the constructed hypergraph can be arbitrarily large. Note also that, in case the problem were \FPT with the combined parameter distance and rank, it would not admit polynomial kernels unless ${\sf NP} \subseteq {\sf coNP}/{\sf poly}$ by our negative result in \autoref{thm:no-kernels}, since the rank of a hypergraph is never larger than its number of vertices.
It is worth mentioning that the hardness result of Cardinal and Steiner on hypergraphic polytopes~\cite[Theorem 4.3]{CardinalS25} holds even if the input hypergraph has rank three (and additionally, bounded maximum degree).

Also, we proved in \autoref{thm:inapproximability} that
  there exists some constant $c > 0$ such that the {\sc Flip Distance Between Acyclic Orientations} problem does not admit a polynomial-time $c \cdot \log (|V|+|{\mathcal E}|)$-approximation algorithm unless $\P = \NP$, where $H=(V,{\mathcal E})$ is the input hypergraph. It remains open whether there exists a polynomial-time approximation algorithm matching (asymptotically) this ratio.

%



\newpage
\bibliography{References}
\end{document}